%% file: ccta26_longVer_V5.tex
\documentclass[conference,onecolumn]{IEEEtran}
\IEEEoverridecommandlockouts

\usepackage{cuted} 
\usepackage{tcolorbox}
\usepackage{cite}
\usepackage{amsmath,amssymb,amsfonts}
\usepackage{mathtools}
\usepackage{bm}
\usepackage{algorithmic}
\usepackage{graphicx}
\usepackage{subcaption}
\usepackage{textcomp}
\usepackage{xcolor}
\usepackage[acronym]{glossaries}
\usepackage{empheq}      
\usepackage{arydshln}
\usepackage[american]{circuitikz}
\usetikzlibrary{fit}

\newtheorem{theorem}{Theorem}
\newtheorem{lemma}{Lemma}
\newtheorem{proposition}{Proposition}
\newtheorem{assumption}{Assumption}
\newtheorem{remark}{Remark}

\newenvironment{proof}{%
  \par\noindent\textbf{Proof.}\ }{%
  \hfill$\blacksquare$\par
}

\newcommand{\R}{\mathbb{R}}

\newacronym{ncs}{CPS}{cyber-physical system}
\newacronym{ftc}{FTC}{fault-tolerant control}
\newacronym{spr}{SPR}{strictly positive real}
\newacronym{dgu}{DGU}{distributed generation unit}
\newacronym{smo}{SMO}{sliding–mode observer}
\newacronym{a-pdgd}{Aug-PDGD}{augmented primal-dual gradient dynamics}
\newacronym{kkt}{KKT}{Karush--Kuhn--Tucker}

\def\BibTeX{{\rm B\kern-.05em{\sc i\kern-.025em b}\kern-.08em
    T\kern-.1667em\lower.7ex\hbox{E}\kern-.125emX}}

\begin{document}

\title{A Primal-Dual-Based Active Fault-Tolerant Control Scheme for Cyber-Physical Systems: Application to DC Microgrids\\
\thanks{This work is partially supported by the Bundesministerium fur Wirtschaft und Klimaschutz (BMWK), Germany under LuFo VI - 2 joint project “Safe and reliable electrical and thermal networks for hybrid-electric drive systems (ETHAN, project number: 20L2103F1”. This research is also supported by the German Federal Government, the Federal Ministry of Research, Technology and Space (BMFTR) and the State of Brandenburg within the framework of the joint project EIZ: Energy Innovation Center (project numbers 85056897 and 03SF0693A) with funds from the Structural Development Act (Strukturstärkungsgesetz) for coal-mining regions.}
}

\author{
\IEEEauthorblockN{
Wasif H. Syed\IEEEauthorrefmark{1},
Juan E. Machado\IEEEauthorrefmark{1},
Johannes Schiffer\IEEEauthorrefmark{1}\IEEEauthorrefmark{4}
}
\IEEEauthorblockA{\IEEEauthorrefmark{1}
\textit{Chair of Control Systems and Network Control Technology,}\\
\textit{Brandenburg University of Technology Cottbus--Senftenberg,}\\
Cottbus, Germany,\\
\{syed,machadom,schiffer\}@b-tu.de.
}
\IEEEauthorblockA{\IEEEauthorrefmark{4}
\textit{Fraunhofer Institute for Energy Infrastructures and Geotechnologies,}\\
 Cottbus, Germany.\\
}
}

\maketitle

\begin{abstract}
We consider the problem of active fault-tolerant control in cyber-physical systems composed of strictly passive linear-time invariant dynamic subsystems. We cast the problem as a constrained optimization problem and propose an augmented primal-dual gradient dynamics-based fault-tolerant control framework that enforces network-level constraints and provides optimality guarantees for the post-fault steady-state operation. By suitably interconnecting the primal-dual algorithm with the cyber-physical dynamics, we provide sufficient conditions under which the resulting closed-loop system possesses a unique and exponentially stable equilibrium point that satisfies the  Karush--Kuhn--Tucker (KKT) conditions of the constrained problem. The framework's effectiveness is illustrated through numerical experiments on a DC microgrid.
\end{abstract}

\begin{IEEEkeywords}
Fault-Tolerant Control, Primal-Dual Dynamics,  Distributed Control, Microgrids, Constraint Optimization, Optimal Control.
\end{IEEEkeywords}

\section{Introduction}

A \gls{ncs} consists of several subsystems coupled through a physical infrastructure and a flexible communication layer~\cite{ferrari}. Such architectures are now widespread in industrial automation (e.g., smart grids, manufacturing, process control, automotive systems, and intelligent highways), motivating extensive research on their modeling, analysis, and control~\cite{Xian-Ming}.
In many classes of \gls{ncs}s the subsystems must act cooperatively, for instance power plants jointly meeting a city’s demand, trucks platooning, or tugboats coordinating maneuvers in ports~\cite{ferrari}. However, conventional feedback control can lead to underperformance or even instability under component malfunctions, which is particularly undesired in safety-critical applications~\cite{ZHANG}. This motivates the development of \gls{ftc} schemes that maintain acceptable performance despite faults and other abrupt changes in  \gls{ncs} dynamics~\cite{Francesca}. 

\gls{ftc} schemes are commonly categorized into passive or active approaches~\cite{Francesca}. On the one hand, passive \gls{ftc} would typically employ a fixed robust controller to tolerate anticipated faults without reconfiguration. However,  its effectiveness can degrade under large or multiple faults. On the other hand, active \gls{ftc} typically entails a controller reconfiguration (or of its parameters) after fault detection (and isolation) in order to recover acceptable performance~\cite{Francesca}. 
In this paper, we focus on active FTC design for a class of \gls{ncs}s in which each subsystem has strictly passive dynamics and in which subsystem interconnections are power-preserving, leading to an overall \gls{ncs} which remains strictly passive, thereby providing a structured basis for \gls{ftc} design \footnote{In many \gls{ncs}s, subsystems are passive either due to inherent physical properties~\cite{van2000l2}  or can be rendered passive via local control~\cite{ortega2004interconnection}.}.

Active \gls{ftc} design has been explored in \cite{ferrari,SchenkCooperative,SchenkMorocco,SchenkFrance} from the viewpoint of flexible task assignment: when a subsystem is degraded, the cooperative objective is redistributed according to residual capabilities, leveraging system-level redundancy without reconfiguring local controllers. Moreover, a decentralized adaptive \gls{ftc} scheme based on backstepping is studied in \cite{ChunXie}, where each subsystem is rendered input-to-state stable with respect to its interconnection inputs and cyclic small-gain conditions are enforced, ensuring network-level stability under actuator faults and disturbances. In \cite{Francesca}, a distributed tube-based MPC is used for active fault isolation/diagnosis in \gls{ncs}: diagnostic inputs are optimized to separate fault models while constraints are enforced locally via invariant tubes; if reconfiguration is infeasible, the affected subsystem is unplugged.

Overall, the above-mentioned FTC design approaches provide valuable mechanisms for fault accommodation (task redistribution, adaptive compensation, or diagnosis). However, they lack a network-level constrained optimization with explicit optimality guarantees that co-optimizes post-fault performance under coupled constraints (e.g., shared resources and interconnection limits), which is central in \gls{ncs}s such as power and energy networks (e.g., microgrids).
To address these limitations, in the present work we make the following  contributions to \gls{ftc} design for (strictly passive) CPSs: 
\begin{itemize}
    \item We cast the post-fault network operation of a \gls{ncs} as a constrained convex optimization problem and formulate an \gls{a-pdgd} for its solution. The \gls{a-pdgd} enables strict constraint satisfaction, and exhibits an energy-/cost-optimal operating point. 
    \item By suitably interconnecting the considered strictly passive \gls{ncs} with the \gls{a-pdgd}, we provide an explicit characterization of the controller parameters under which the resulting closed-loop system  
    admits an exponentially stable equilibrium point that satisfies the  \gls{kkt} conditions of the constrained problem. This is accomplished by combining prior results in \cite{qu_li_augPD} that establish exponential stability of the equilibrium of the \gls{a-pdgd} with the control-by-interconnection (CbI) framework \cite{CBI_OrtegaArjan}. 
    \item We demonstrate the potential of the proposed \gls{ftc} for \gls{ncs}s via an application to class of clustered DC microgrids. This is particularly appealing scenario, since following a capacity degradation or outage, the \gls{ftc} based on the \gls{a-pdgd} computes the least-cost, loss-aware redistribution of generation while respecting practical limits on voltages and currents.
\end{itemize}

\subsection*{Notation}

The notation used in this paper is as follows. For elements \(x_i\in\R\), the expression \(x = \mathrm{col}(x_i)\) denotes the column vector formed by stacking \(x_i\). The operators \(\mathrm{diag}(x_i)\) and \(\mathrm{blkdiag}(x_i)\) denote diagonal and block-diagonal matrices with entries or blocks \(x_i\), respectively. The symbols \(\mathbf{I}\) and \(\mathbf{0}\) denote the identity and zero matrices of appropriate dimensions. For any signal \(x:\R\to\R^n\), the deviation from a constant \(\bar{x}\in\R^n\) is written as \(\tilde{x} = x - \bar{x}\). For any symmetric matrix \(M \in \mathbb{R}^{n\times n}\), the quadratic form \(x^\top M x\) is compactly written as \(\|x\|_{M}^{2}\). For two symmetric matrices \(A\) and \(B\), the notation \(A \succ B\) (respectively \(A \succeq B\))
means that \(A - B\) is positive definite (respectively positive semidefinite). For a square matrix $M$, $\lambda_i(M)$ denotes its $i$-th eigenvalue, and $\text{Re}(\lambda_i(M))$ denotes its real part. The symbol \(\lambda_{\max}(M)\) denotes the largest eigenvalue of a symmetric matrix \(M\). For any collection of matrices $M_{i,j}$, the notation 
$[\,M_{i,j}\,]_{i,j=1}^N$ denotes the block matrix obtained by placing $M_{i,j}$ in the $(i,j)$-th block position. For any set $S\subset\mathbb{R}^n$, $\operatorname{conv}(S)$ denotes its convex hull.

\section{CPS Model} \label{sec:SystemModel}

\subsection{Subsystem dynamics and interconnection rule}

We consider a \gls{ncs} composed of \(N>1\) interconnected LTI subsystems with dynamics
\begin{subequations}\label{eq:plant_i}
\begin{align}
\dot{x}_i &= A_i x_i + B_i u_i + d_i + G_i w_i,\\
y_i &= C_i x_i,\\
z_i &= E_i x_i,
\end{align}
\end{subequations}
where \(x_i \in \mathbb{R}^{n_i}\) is the state, \(u_i \in \mathbb{R}^{m_i}\) is the control input,
\(y_i \in \mathbb{R}^{p_i}\) is the measured output, \(z_i \in \mathbb{R}^{s_i}\) is the interconnection output, \(w_i \in \mathbb{R}^{q_i}\) is the interconnection input, and \(d_i \in \mathbb{R}^{n_i}\) is a constant vector, $i=1,\ldots,N$. 
The interconnection of the $i$-th subsystem with other subsystems is given by 
\begin{equation}
w_i = \sum_{j=1}^N \Omega_{i,j} z_j,
\label{wi}    
\end{equation}
with constant matrices \(\Omega_{i,j} \in \mathbb{R}^{q_i \times s_j}\), which are only nonzero if subsystem $i$ is connected to subsystem $j$ via a physical link.

We make the following technical assumptions on the dynamics \eqref{eq:plant_i}, \eqref{wi}, $i=1,\ldots,N$.
\begin{assumption}\label{assump:local_SPR}
For each \(i=1,\ldots,N\), the triples \((A_i,B_i,C_i)\) and \((A_i,G_i,E_i)\) in \eqref{eq:plant_i} are minimal, and the transfer function matrices
\[
G_{p,i}(s)=C_i(s \mathbf{I} -A_i)^{-1}B_i~~\text{and}~~
G_{z,i}(s)=E_i(s\mathbf{I}-A_i)^{-1}G_i
\]
are \gls{spr}, where $s \in \mathbb{C}$ is the Laplace variable.

\hfill $\square$
\end{assumption}

\begin{assumption}\label{assump:interconnection_passive}
The interconnection \eqref{wi} of any two subsystems \eqref{eq:plant_i} is power-preserving, i.e.,
\[
\Omega_{i,j} + \Omega_{j,i}^\top = \mathbf{0}\quad \forall i=1,\ldots,N,\quad \forall j=1\ldots,N,\quad j\neq i.
\]
\hfill $\square$
\end{assumption}

Assumption~\ref{assump:local_SPR} ensures that each subsystem is internally stable, i.e., \(A_i\) is Hurwitz, and strictly passive in the input--output pairs \((u_i,y_i)\) and \((w_i,z_i)\) \cite{khalil2002nonlinear} \footnote{In practice, this may result from intrinsic plant properties or from a local pre-stabilizing controller, in which case \(u_i\) acts as a supervisory input to achieve objectives other than stability.}. Assumption~\ref{assump:interconnection_passive} guarantees that the coupling does not inject energy into the network and preserves passivity when subsystems are interconnected \cite{van2000l2}.

\subsection{Compact \gls{ncs} model}

Considering \eqref{eq:plant_i} and \eqref{wi}, for $i=1,\ldots,N$, let
\(x = \mathrm{col}(x_1,\dots,x_N) \in \mathbb{R}^{n}\),
\(u = \mathrm{col}(u_1,\dots,u_N) \in \mathbb{R}^{m}\),
\(y = \mathrm{col}(y_1,\dots,y_N) \in \mathbb{R}^{p}\),
\(z = \mathrm{col}(z_1,\dots,z_N) \in \mathbb{R}^{s}\),
and \(d = \mathrm{col}(d_1,\dots,d_N) \in \mathbb{R}^{n}\).
Moreover, let
\(A=\mathrm{blkdiag}(A_i),\;
B=\mathrm{blkdiag}(B_i),\;
C=\mathrm{blkdiag}(C_i),\; E=\mathrm{blkdiag}(E_i),\;
G=\mathrm{blkdiag}(G_i),\)
and \(\Omega=[\Omega_{i,j}]_{i,j=1}^N\).
By noting that \(w=\Omega z\) and \(z=Ex\),
then the overall \gls{ncs} model \eqref{eq:plant_i}, \eqref{wi}, $i=1,\ldots,N,$ can be written compactly as
\begin{equation}\label{eq:CompactModel}
\Sigma_{\mathrm P}:\;
\begin{cases}
\dot{x} = A x + B u + d + G\Omega E x,\\
y = C x,\\
z = E x.
\end{cases}
\end{equation}

We make the following assumption on the existence of a desired equilibrium point of the uncontrolled \gls{ncs} \eqref{eq:CompactModel}.

\begin{assumption}
With \(u =u_i^*= \mathbf{0}\), the interconnected system \eqref{eq:CompactModel} admits an exponentially stable equilibrium point
\(x^*\) with output \(y^*\). At equilibrium, each subsystem satisfies the  constraints $x_i^*\in \mathcal{X}_i$, $y_i^*\in \mathcal{Y}_i$ and $u_i^*\in \mathcal{U}_i$, where
\begin{subequations}\label{eq:calX_i_calYi_calUi}
\begin{align}
\mathcal{X}_i &= \{\, x_i\in \mathbb{R}^{n_i} \mid R_{\mathrm{eq},x,i} x_i = b_{x,i},\;
    R_{\mathrm{ineq},x,i} x_i \le h_{x,i} \,\},\\
\mathcal{Y}_i &= \{\, y_i\in\mathbb{R}^{p_i} \mid R_{\mathrm{eq},y,i} y_i = b_{y,i},\;
    R_{\mathrm{ineq},y,i} y_i \le h_{y,i} \,\},\\
\mathcal{U}_i &= \{\, u_i\in \mathbb{R}^{m_i} \mid R_{\mathrm{eq},u,i} u_i = b_{u,i},\;
    R_{\mathrm{ineq},u,i} u_i \le h_{u,i} \,\}
\end{align}
\end{subequations}
are convex sets \footnote{ $R_{\mathrm{eq},x,i}\in\mathbb{R}^{r_{x,i}\times n_i}$, $b_{x,i}\in\mathbb{R}^{r_{x,i}}$,
$R_{\mathrm{ineq},x,i}\in\mathbb{R}^{k_{x,i}\times n_i}$, $h_{x,i}\in\mathbb{R}^{k_{x,i}}$,
$R_{\mathrm{eq},y,i}\in\mathbb{R}^{r_{y,i}\times p_i}$, $b_{y,i}\in\mathbb{R}^{r_{y,i}}$,
$R_{\mathrm{ineq},y,i}\in\mathbb{R}^{k_{y,i}\times p_i}$, $h_{y,i}\in\mathbb{R}^{k_{y,i}}$,
$R_{\mathrm{eq},u,i}\in\mathbb{R}^{r_{u,i}\times m_i}$, $b_{u,i}\in\mathbb{R}^{r_{u,i}}$,
$R_{\mathrm{ineq},u,i}\in\mathbb{R}^{k_{u,i}\times m_i}$, $h_{u,i}\in\mathbb{R}^{k_{u,i}}$ are user-defined constraint parameters.
}.  
\hfill $\square$
\end{assumption}

The following lemma states that the  CPS model  \eqref{eq:CompactModel} is strictly shifted-passive.  This property is actively exploited for the subsequent  FTC design and corresponding closed-loop stability analysis.

\begin{lemma}\label{lemma:PlantPassivity}
Fix $0<\tau_1<2\min_i\{-\text{Re}(\lambda_i(A_\mathrm{p}))\}$. With $A_{\mathrm p}=A+G\Omega E$, there is a unique solution $P_1\succ 0$ to the Lyapunov equation
\begin{equation}\label{eq:P1_lyap}
A_{\mathrm p}^\top P_1 + P_1 A_{\mathrm p} = -\tau_1 P_1,
\end{equation}
and the plant dynamics \eqref{eq:CompactModel} are strictly shifted-passive with quadratic storage function
$S_{\mathrm p}(\tilde{x})= \tfrac12 \|\tilde{x}\|_{P_1}^2$ and with respect to the input--output pair $(\tilde{u},\tilde{y})$.
In particular, it holds that
\begin{equation}\label{eq:dissipation_inequality_PLANT}
\dot{S}_{\mathrm{p}}(\tilde{x}) = - \tau_1 \tilde{x}^\top P_1 \tilde{x} + \tilde{y}^\top \tilde{u},
\end{equation}
along any system trajectory.
\hfill $\square$
\end{lemma}

\begin{proof}
Assumptions~\ref{assump:local_SPR} and \ref{assump:interconnection_passive}
imply that \(\Sigma_{\mathrm P}\) in \eqref{eq:CompactModel} is a power-preserving interconnection of
\gls{spr} subsystems. Then, the transfer function matrix $G_\mathrm{p}(s)$ of the error dynamics, given by
\[
G_{\mathrm p}(s)= C (s\mathbf{I} - A_{\mathrm p})^{-1} B, 
\quad
\]
is \gls{spr}.  This implies in particular that \(A_{\mathrm p}\) is Hurwitz (see \cite[Lemmas~6.3 and 6.4]{khalil2002nonlinear}). Then there exists a unique matrix $P_1=P_1^\top>0$ that solves \eqref{eq:P1_lyap}. It is straightfoward to verify that, along system trajectories, the time-derivative of the storage function $S_\mathrm{p}$ satisfies \eqref{eq:dissipation_inequality_PLANT}.

\end{proof}

\section{Problem Statement and Solution Approach} \label{sec:ProbForm}
\subsection{Problem Statement}
Assume that at time \(t_\mathrm{f}\), the subsystem \(i_\mathrm{f}\) experiences a
fault (e.g., actuator failure or an additive disturbance), leading to the
reconfigured dynamics of \eqref{eq:plant_i}, \eqref{wi}, for the $i_\mathrm{f}$-th subsystem:
\begin{subequations}\label{eq:plant_i_f}
\begin{align}
\dot{\hat x}_{i_\mathrm{f}}
&= \hat A_{i_\mathrm{f}} \hat x_{i_\mathrm{f}}
 + \hat B_{i_\mathrm{f}} \hat u_{i_\mathrm{f}}
 + \hat d_{i_\mathrm{f}}
 + \hat G_{i_\mathrm{f}} \hat w_{i_\mathrm{f}},\\
\hat y_{i_\mathrm{f}} &= \hat C_{i_\mathrm{f}} \hat x_{i_\mathrm{f}},\\
\hat z_{i_\mathrm{f}} &= \hat E_{i_\mathrm{f}} \hat x_{i_\mathrm{f}},
\end{align}
with interconnection
\[
\hat w_{i_\mathrm{f}}
= \hat\Omega_{i_\mathrm{f},i_\mathrm{f}}\hat z_{i_\mathrm{f}}
  + \sum_{j\neq i_\mathrm{f}} \hat\Omega_{i_\mathrm{f},j} z_j,
\]
\end{subequations}
and with all variables and matrices  defined analogously to~\eqref{eq:plant_i}\footnote{In particular,
$\hat x_{i_\mathrm{f}}\!\in\mathbb{R}^{n_{i_\mathrm{f}}},\;
 \hat u_{i_\mathrm{f}}\!\in\mathbb{R}^{m_{i_\mathrm{f}}},\;
 \hat y_{i_\mathrm{f}}\!\in\mathbb{R}^{p_{i_\mathrm{f}}},\;
 \hat z_{i_\mathrm{f}}\!\in\mathbb{R}^{s_{i_\mathrm{f}}},\;
 \hat w_{i_\mathrm{f}}\!\in\mathbb{R}^{q_{i_\mathrm{f}}},
\hat d_{i_\mathrm{f}}\!\in\mathbb{R}^{n_{i_\mathrm{f}}},\;
 \hat A_{i_\mathrm{f}}\!\in\mathbb{R}^{n_{i_\mathrm{f}}\times n_{i_\mathrm{f}}},\;
 \hat B_{i_\mathrm{f}}\!\in\mathbb{R}^{n_{i_\mathrm{f}}\times m_{i_\mathrm{f}}}, 
\hat G_{i_\mathrm{f}}\!\in\mathbb{R}^{n_{i_\mathrm{f}}\times q_{i_\mathrm{f}}},\;
 \hat C_{i_\mathrm{f}}\!\in\mathbb{R}^{p_{i_\mathrm{f}}\times n_{i_\mathrm{f}}},\;
 \hat E_{i_\mathrm{f}}\!\in\mathbb{R}^{s_{i_\mathrm{f}}\times n_{i_\mathrm{f}}}, \hat\Omega_{i_\mathrm{f},j}\in\mathbb{R}^{q_{i_\mathrm{f}}\times s_j}.$}.
We assume that Assumptions~\ref{assump:local_SPR} and~\ref{assump:interconnection_passive}
remain valid for the faulty subsystem~\eqref{eq:plant_i_f}. Moreover, we consider that a fault may produce either one, or a combination, of the following situations:
\begin{itemize}
\item \textbf{Steady-state error:}
      \(\bar x_i\neq x_i^*\) or \(\bar{\hat x}_{i_\mathrm{f}}\neq x_{i_\mathrm{f}}^*\).
\item \textbf{Output error:}
      \(\bar y_i\neq y_i^*\) or \(\bar{\hat y}_{i_\mathrm{f}}\neq \hat y_{i_\mathrm{f}}^*\).
\item \textbf{Constraint violation:}
      \(\bar x_i\notin\mathcal{X}_i\),
      \(\bar{\hat x}_{i_\mathrm{f}}\notin\hat{\mathcal{X}}_{i_\mathrm{f}}\),
      \(\bar y_i\notin\mathcal{Y}_i\),
      \(\bar{\hat y}_{i_\mathrm{f}}\notin\hat{\mathcal{Y}}_{i_\mathrm{f}}\) \footnote{To simplify the notation, we henceforth drop the use of the symbol \(\hat{(\cdot)}\) for the faulty subsystem \eqref{eq:plant_i_f}.}.
\end{itemize}
Then, the problem we address in this paper is that of designing a dynamic, feedback control law for the control input $u$ that   computes a new feasible operating point that minimizes the error with respect to the pre-fault equilibrium---while satisfying all post-fault physical and operational constraints---and that  renders such an optimal equilibrium exponentially stable. 
Formally, let $J:\mathbb{R}^{n}\times\mathbb{R}^{p}\times\mathbb{R}^{m}\to\mathbb{R}$,  denote a cost function penalizing the error between a post-fault equilibrium triple $(\bar x, \bar y, \bar u)$ and the pre-fault desired equilibrium  \((x^*,y^*,u^*)\). Then our goal is to design a \gls{ftc} scheme consisting of a continuous controller for \(u=\mathrm{col}(u_i)\), such that the \gls{ncs} trajectories remain bounded and \(x \to \bar x = \mathrm{col}(\bar x_i)\), exponentially fast, as $t\to\infty$, where \(\bar x\) is the solution of the following steady-state optimization problem:
\begin{equation}\label{eq:OptimizationProblem}
\begin{aligned}
\min_{\bar x,\bar y,\bar u}\;& J(\bar x,\bar y,\bar u)\\
\text{s.t.}\quad
0 &= (A+G\Omega E)\bar x + B\bar u + d,\\
\bar y &= C\bar x,\\
\bar x &\in\mathcal{X},\;
\bar y\in\mathcal{Y},\;
\bar u\in\mathcal{U},
\end{aligned}
\end{equation}
where  \(\mathcal{X}=\prod_{i=1}^N \mathcal{X}_i\),
\(\mathcal{Y}=\prod_{i=1}^N \mathcal{Y}_i\),
and \(\mathcal{U}=\prod_{i=1}^N \mathcal{U}_i\), with $\mathcal{X}_i$, $\mathcal{Y}_i$ and $\mathcal{U}_i$ defined in \eqref{eq:calX_i_calYi_calUi}.  A simple choice for $J$ is a quadratic function of the form
\[
J(\bar x, \bar y, \bar u)=\tfrac12 \|\bar{x}-x^*\|_{K_x}^{2} + \tfrac12 \|\bar{y}-y^*\|_{K_y}^{2}+ \tfrac12 \|\bar{u}-u^*\|_{K_u}^{2} ,
\]
with \(K_x\succ 0\), \(K_y\succ 0\), and \(K_u\succ 0\). Nonetheless, in our analysis other type of cost functions can be considered as long as they satisfy a strong convexity assumption (see Assumption~\ref{assump:uConvex} in Section~\ref{sec:Controller} for details.)

\subsection{Solution Approach}

Following \cite{claudio_arjan,michele}, the approach we employ to solve the considered control problem consists in formulating a continuous-time   Aug-PDGD associated to \eqref{eq:OptimizationProblem}, and, after doting it with suitable input channels, use it as a dynamic feedback controller in a CbI~\cite{CBI_OrtegaArjan} fashion. Unlike \cite{claudio_arjan,michele}, we employ the augmented version rather than the standard PDGD, leveraging the theoretical and numerical benefits discussed in \cite{qu_li_augPD}. Moreover, the considered CbI-based design differs from the time-scale separation approach in \cite{Dall’Anese} in the sense that our scheme actively exploits both the plant's and the augmented PDGD's strict shifted-passivity properties, which enables the explicit construction of a Lyapunov function for establishing exponential stability of the closed-loop system equilibrium \footnote{Other benefits of using a dynamic controller based on the Aug-PDGD, instead of using a static feedforward input $u=\bar u$ that solves \eqref{assump:uConvex}, is that it yields a feedback-optimization architecture that adapts online, improving robustness and performance and enabling real-time adjustment of $\bar u$ under disturbances~\cite{opt-alg-robust-fb}.}.

\section{Proposed Active Primal-Dual FTC Scheme} \label{sec:Controller}

\subsection{Derivation of the Aug-PDGD}

For deriving the Aug-PDGD, we begin by reformulating  the optimization problem \eqref{eq:OptimizationProblem} as follows:
\begin{subequations}\label{eq:FTCObjectives3}
\begin{align}
\begin{split}
\min_{\bar{\xi}} \quad & J (\bar{\xi})     \\ 
\text{subject to} \quad 
& R_{\text{eq}}\,\bar{\xi}  = b, \\ 
& R_{\text{ineq}}\,\bar{\xi}  \leq h,
\end{split}
\end{align}
where 
\begin{equation}
\bar{\xi}=\text{col}(\bar x,\bar y,\bar u)\in \mathbb{R}^{n_\xi=n+p+m},    
\end{equation}
\begin{align}\label{eq:R_block_compact}
R_{\mathrm{eq}} &=
\begin{bmatrix}
    -C & \mathbf{I} & \mathbf{0} \\
    A + G\Omega E & \mathbf{0} & B \\
    R_{\mathrm{eq},x} & \mathbf{0} & \mathbf{0} \\
    \mathbf{0} & R_{\mathrm{eq},y} & \mathbf{0} \\
    \mathbf{0} & \mathbf{0} & R_{\mathrm{eq},u}
\end{bmatrix} \in \mathbb{R}^{n_{\text{eq}}\times n_\xi}, \\
% \qquad
R_{\mathrm{ineq}} &=
\begin{bmatrix}
    R_{\mathrm{ineq},x} & \mathbf{0} & \mathbf{0} \\
    \mathbf{0} & R_{\mathrm{ineq},y} & \mathbf{0} \\
    \mathbf{0} & \mathbf{0} & R_{\mathrm{ineq},u}
\end{bmatrix}\in \mathbb{R}^{n_{\text{ineq}}\times n_\xi},
\\
b & \coloneqq \mathrm{col}(\mathbf{0},\,-d,\,b_x,\,b_y,\,b_u)\in \mathbb{R}^{n_{\text{eq}}},\\
h & \coloneqq \mathrm{col}(h_x,\,h_y,\,h_u)\in \mathbb{R}^{n_{\text{ineq}}},
\end{align}
\end{subequations}
with $n_{\text{eq}}=p+n+r_x+r_y+r_u$, $n_{\text{ineq}}= k_x+k_y+k_u$,  $R_{\mathrm{eq},x} \coloneqq
\mathrm{blkdiag}(R_{\mathrm{eq},x,1},\dots,R_{\mathrm{eq},x,N})$, $b_x \coloneqq \mathrm{col}(b_{x,1},\dots,b_{x,N})$, $R_{\mathrm{ineq},x} \coloneqq
\mathrm{blkdiag}(R_{\mathrm{ineq},x,1},\dots,R_{\mathrm{ineq},x,N})$, and $h_x \coloneqq \mathrm{col}(h_{x,1},\dots,h_{x,N})$.
The remaining block matrices and vectors, namely $R_{\mathrm{eq},y}$, $R_{\mathrm{eq},u}$, $R_{\mathrm{ineq},y}$, $R_{\mathrm{ineq},u}$, $b_y$, $b_u$, $h_y$ and $h_u$ are defined analogously.

The following technical assumptions about the optimization problem~\eqref{eq:FTCObjectives3} are instrumental for showing the existence of a unique solution to \eqref{eq:FTCObjectives3}, as well as for establishing  strict shifted-passivity properties for the corresponding Aug-PDGD with respect to a  suitable input-output pair.

\begin{assumption} \label{assump:uConvex}
    The cost function $J$ in \eqref{eq:OptimizationProblem} is twice continuously differentiable, $\mu$-strongly convex, and $\ell$-smooth in its arguments, for some $0<\mu\leq \ell$. 
\hfill $\square$
\end{assumption}    
Assumption~\ref{assump:uConvex} implies that, in particular, for all $\xi_1,\xi_2\in\mathbb{R}^{n+m+p}$, it holds that
\begin{equation*}
\mu\|\xi_1-\xi_2\|^2 \;\le\;
\big\langle \nabla J(\xi_1)-\nabla J(\xi_2),\, \xi_1-\xi_2 \big\rangle
\;\le\; \ell \|\xi_1-\xi_2\|^2,
\end{equation*}
or, equivalently, $\mu \mathbf{I} \preceq \nabla^2 J(\xi) \preceq \ell \mathbf{I}$ for all $\xi\in\mathbb{R}^{n+m+p} $.

\begin{assumption}\label{assump:BoundOnR}
Let $n_c \coloneqq n_{\mathrm{eq}} + n_{\mathrm{ineq}}$. The stacked constraint matrix
\begin{align}
R \;=\;
\begin{bmatrix}
R_{\mathrm{eq}} \\[2pt]
R_{\mathrm{ineq}}
\end{bmatrix}
    \label{eq:R_stacked}
\end{align}
has full row rank.
\hfill $\square$
\end{assumption}
With Assumption~\ref{assump:BoundOnR}, there exist constants $0 < \kappa_1 \le \kappa_2$ such that
\begin{align}
\kappa_1 \mathbf{I} \;\le\; R R^{\top} \;\le\; \kappa_2 \mathbf{I},    \label{eq:BoundOnR}
\end{align}
where $RR^{\top}\in\mathbb{R}^{n_c\times n_c}$.

In view of Assumption~\ref{assump:uConvex} and the fact that the constraints are affine in the decision variables,  the optimization problem~\eqref{eq:FTCObjectives3} is convex. Then, it is solvable if and only if the following \gls{kkt} conditions~\cite{boyd_book} are solvable:
\begin{align}\label{eq:KktConditions}
    \begin{split}
        \nabla J(\bar{\xi}) + R_{\text{eq}}^\top \bar{\nu}_{\text{eq}} + R_{\text{ineq}}^\top \bar{\nu}_{\text{ineq}} &= \mathbf{0},\\
        R_{\text{eq}}\,\bar{\xi}  - b &= \mathbf{0} , \\ 
        R_{\text{ineq}}\,\bar{\xi} - h &\leq \mathbf{0},\\
        \bar{\nu}_{\text{ineq}}  &\geq \mathbf{0}, \\       \bar{\nu}_{\text{ineq}}^\top\!\big(R_{\text{ineq}}\,\bar{\xi} - h\big) &= \mathbf{0},
    \end{split}
\end{align}
where $\bar{\nu}_{\text{eq}}  \in \mathbb{R}^{n_{\text{eq}}}$ and $\bar{\nu}_{\text{ineq}} \in \mathbb{R}^{n_{\text{ineq}}}_{\geq 0}$ are dual variables. Under Assumptions~\ref{assump:uConvex} and~\ref{assump:BoundOnR}, the \gls{kkt}
conditions~\eqref{eq:KktConditions} admit a unique solution \cite{boyd_book}.
In order to define the Aug-PDGD associated to \eqref{eq:FTCObjectives3} and \eqref{eq:KktConditions}, we  follow \cite{qu_li_augPD} and introduce the  augmented Lagrangian function 
\begin{align}
\mathcal{L}_\rho(\bar \xi,\bar \nu_{\text{eq}},\bar \nu_{\text{ineq}})
&= J(\bar \xi)
+ \bar \nu_{\text{eq}}^\top(R_{\text{eq}}\bar \xi - b)
\nonumber \\
& + \sum_{i=1}^{n_{\text{ineq}}} 
H_\rho\Big( (R_{\text{ineq}}\bar \xi - h)_i,\; (\bar \nu_{\text{ineq}})_i \Big),
\label{eq:AugLag_FTC}
\end{align}
where $H_\rho$ is defined as
\begin{equation}\label{eq:Hrho_FTC} 
H_\rho (a,b)  =\begin{cases} ab+\frac{\rho}{2}a^2 & \rho a+b\ge 0, \\[6pt] -\frac{b^2}{2\rho}, & \rho a+b< 0, \end{cases}
\end{equation}
with $\rho>0$ being a free parameter.
Then, the Aug-PDGD is defined as follows~\cite{qu_li_augPD}:
\begin{subequations}\label{eq:primal_dual_NON_EXPLICIT}
\begin{align}
        \dot \xi & = -\nabla_\xi \mathcal{L}(\xi,\nu_\mathrm{eq},\nu_\mathrm{ineq}),\\
    \dot \nu_\mathrm{eq} & =  \eta \nabla_{\nu_\mathrm{eq}} \mathcal{L}(\xi,\nu_\mathrm{eq},\nu_\mathrm{ineq}),\\
        \dot \nu_\mathrm{ineq} & = \eta
        \nabla_{\nu_\mathrm{ineq}} \mathcal{L}(\xi,\nu_\mathrm{eq},\nu_\mathrm{ineq}),
\end{align}
\end{subequations}
where $\xi(t)=\text{col}(\hat{x}(t),\hat{y}(t),\hat{u}(t))$, $\nu_\mathrm{eq}(t)$ and $\nu_\mathrm{ineq}(t)$ are the virtual states  associated to the primal ($\bar \xi=\text{col}(\bar x,\bar y, \bar u)$) and dual variables ($\bar \nu_\mathrm{eq}$, $\bar \nu_\mathrm{ineq})$,  and $\eta>0$ is a design parameter.  Note that under Assumptions~\ref{assump:uConvex} and~\ref{assump:BoundOnR}, the Aug-PDGD \eqref{eq:primal_dual_NON_EXPLICIT} admits a unique equilibrium, which is   in one-to-one correspondence with the unique solution of the \gls{kkt} conditions~\eqref{eq:KktConditions}.  

We underscore that in~\cite{qu_li_augPD}, a PDGD with affine equality constraints and an \gls{a-pdgd} with affine inequality constraints are analyzed separately, and it is noted that they can be combined. Building on their results, here we consider an Aug-PDGD that deals with affine equality and inequality constraints simultaneously.  

To write \eqref{eq:primal_dual_NON_EXPLICIT} in explicit form, note that
\begin{subequations}
    \begin{align*}
   \nabla_\xi \mathcal{L}(\xi,\nu_\mathrm{eq},\nu_\mathrm{ineq}) & =      \nabla_\xi J(\xi)+R_\mathrm{eq}^\top \nu_\mathrm{eq}\\
   & \phantom{=}+ \sum_{i=1}^{n_\mathrm{ineq}} \nabla_{\xi} H_\rho\big((R_{\text{ineq}}\xi - h)_i,(\nu_{\text{ineq}})_i\big),\\
 \nabla_{\nu_\mathrm{eq}} \mathcal{L}(\xi,\nu_\mathrm{eq},\nu_\mathrm{ineq}) & = R_\mathrm{eq}\xi-b,\\
 \nabla_{\nu_\mathrm{ineq}} \mathcal{L}(\xi,\nu_\mathrm{eq},\nu_\mathrm{ineq}) & =  \sum_{i=1}^{n_\mathrm{ineq}} \nabla_{\nu_{\text{ineq}}} H_\rho\big((R_{\text{ineq}}\xi - h)_i,(\nu_{\text{ineq}})_i\big).
    \end{align*}
\end{subequations}
%%%
To  compute the remaining terms with derivatives of $H_\rho$, we introduce, as in \cite{qu_li_augPD}, the vector
\begin{align}\label{eq:vector_g}
g(\xi,\nu_{\text{ineq}})
\;\coloneqq\;
\max\!\Big(
\nu_{\text{ineq}} + \rho(R_{\text{ineq}}\xi - h),\; 0
\Big), 
\end{align}
where the $\max(\cdot)$ is taken component-wise. Then, 
\begin{center}
\scalebox{0.95}{$
\begin{aligned}
 \sum_{i=1}^{n_\mathrm{ineq}} \nabla_{\xi} H_\rho\big((R_{\text{ineq}}\xi - h)_i,(\nu_{\text{ineq}})_i\big)
&= R_{\text{ineq}}^\top g(\xi,\nu_{\text{ineq}}),
\\[2mm]
 \sum_{i=1}^{n_\mathrm{ineq}} \nabla_{\nu_{\text{ineq}}} H_\rho\big((R_{\text{ineq}}\xi - h)_i,(\nu_{\text{ineq}})_i\big)
&= \frac{1}{\rho}\big(g(\xi,\nu_{\text{ineq}}) - \nu_{\text{ineq}}\big).
\end{aligned}
$}
\end{center}
%%%%
Consequently, the Aug-PDGD attains the following form:
\begin{subequations}\label{eq:PdController}
\begin{equation}
   \dot \theta=f(\theta),\quad \theta=\mathrm{col}({\xi},\,{\nu}_{\mathrm{eq}},\,{\nu}_{\mathrm{ineq}})\in \mathbb{R}^{n_\theta=n_\xi+n_\mathrm{eq}+n_\mathrm{ineq}},
\end{equation}
with
\begin{equation}
f(\theta)=\begin{pmatrix}
-\nabla J({\xi})
   - R_{\text{eq}}^\top {\nu}_{\text{eq}}
   - R_{\text{ineq}}^\top g\big({\xi},{\nu}_{\text{ineq}}\big)\\
   \eta (R_{\text{eq}}{\xi} - b)\\
   \frac{\eta}{\rho}\Big(
    g\big({\xi},\nu_{\text{ineq}}\big)
    - {\nu}_{\text{ineq}}
\Big)
\end{pmatrix}.
\end{equation}
\end{subequations}

\subsection{Aug-PDBD-based CbI Framework}

To apply the CbI framework, we first equip the \gls{a-pdgd} \eqref{eq:PdController} with an auxiliary input $v_\mathrm{pd}$ as follows:
\begin{align}\label{eq:PdControllerWithInput}
\Sigma_{\mathrm c}:\quad
\dot{\theta} = f(\theta) + B_{\mathrm{pd}} v_{\mathrm{pd}},
\end{align}
where  \(B_{\mathrm{pd}}\) is the 
associated input matrix (to be defined). Next, we move on to show that \eqref{eq:PdControllerWithInput} is strictly shifted-passive with quadratic storage function, and with respect to a suitable output. In the sequel, the passive input-output pairs of both the CPS model \eqref{eq:CompactModel} and the Aug-PDGD \eqref{eq:PdControllerWithInput} will be used to define a power-preserving interconnection between the two systems.

To establish the shifted-passivity property for \(\Sigma_{\mathrm c}\), we consider the quadratic storage function (candidate) \(S_{\mathrm{c}} = \tilde{\theta}^{\top} P_2 \tilde{\theta}\), where $P_2$ is a symmetric matrix defined as
\begin{equation}\label{eq:matrix_P_2}
 P_2=   \begin{bmatrix}
\eta c\, \mathbf{I} & \eta R_{\mathrm{eq}}^{\top} & \eta R_{\mathrm{ineq}}^{\top} \\
\eta R_{\mathrm{eq}} & c\, \mathbf{I} & \mathbf{0}\\
\eta R_{\mathrm{ineq}} & \mathbf{0} & c\, \mathbf{I}
\end{bmatrix}\in\mathbb{R}^{n_\theta \times n_\theta},
\end{equation}
with $R_{\mathrm{eq}}$ and $R_{\mathrm{ineq}}$ as in \eqref{eq:R_block_compact} and 
\begin{subequations}\label{eq:ChoiceOfC_all}
\begin{align}
c &= \max\!\bigl(c_1,\,c_2,\,c_3\bigr), \label{eq:ChoiceOfC}\\[1mm]
c_1 &\ge
\begin{cases}
\displaystyle \eta\,\sqrt{\frac{\kappa_{2}}{(\eta - \epsilon)(1-\epsilon)}}, 
& \eta \ge 1,\\[3mm]
\displaystyle \sqrt{\frac{\eta\,\kappa_{2}}{(1 - \epsilon)(1-\eta \epsilon)}}, 
& 0 < \eta < 1,
\end{cases}
\label{eq:ChoiceOfC1}\\[1mm]
c_2 &\ge \kappa_2 \rho, \label{eq:ChoiceOfC2}\\[1mm]
c_3 &\ge
20\,\ell\,
\Big[\max\!\left(\tfrac{\rho \kappa_{2}}{\mu},\, \tfrac{\ell}{\mu}\right)\Big]^{2}
\Big[\max\!\left(\tfrac{\eta}{\ell\rho},\, \tfrac{\ell}{\mu}\right)\Big]^{2}
\tfrac{\kappa_{2}}{\kappa_{1}},
\label{eq:ChoiceOfC3}
\end{align}
\end{subequations}
with $0<\epsilon<1$ being a free parameter, and where $\ell$ and $\mu$ are the smoothness and strong-convexity constants of $J$ from Assumption~\ref{assump:uConvex}, $\kappa_1,\kappa_2$ are the spectral bounds of $RR^\top$ from Assumption~\ref{assump:BoundOnR}, and $\rho>0$ is the free parameter introduced in \eqref{eq:Hrho_FTC}. 

\begin{lemma}\label{lemma:PdPassivity}
Fix $\eta>0$ and $\epsilon\in(0,1)$,  let Assumptions~\ref{assump:uConvex} and \ref{assump:BoundOnR} hold, and  fix $c>0$ according to~\eqref{eq:ChoiceOfC_all}. Then,  $P_2$ in \eqref{eq:matrix_P_2} satisfies
\begin{equation}\label{eq:lower_bound_P_2}
P_2 \;\succeq\; \epsilon\,c\,\min(\eta,1)\,\mathbf I.
\end{equation}
Moreover, the \gls{a-pdgd}~\eqref{eq:PdControllerWithInput} is strictly shifted-passive   with quadratic storage function \(S_{\mathrm{c}}(\tilde{\theta})=\tilde{\theta}^{\top}P_2\tilde{\theta}\) and with respect to the input--output pair
\((\tilde{v}_{\mathrm{pd}},\, 2\,B_{\mathrm{pd}}^{\top} P_2 \tilde{\theta})\).
In particular,  it holds that
\[
\dot{S}_{\mathrm{c}}
\le -\tau_2\, \tilde{\theta}^{\top} P_2\, \tilde{\theta}
+ 2\,\tilde{\theta}^{\top} P_2 B_{\mathrm{pd}} \tilde{v}_{\mathrm{pd}}
\]
along any system trajectory, with 
\begin{align}
\tau_2 &\coloneqq \frac{\eta \kappa_1}{2c}.
\label{eq:tau2_def}    
\end{align}
\hfill $\square$
\end{lemma}

\begin{proof}
See Appendix \ref{ProofLemma:PdPassivity}. 
\end{proof}

In view of Lemmas~\ref{lemma:PlantPassivity} and \ref{lemma:PdPassivity},  we are in position to specify the interconnection law between the 
\gls{ncs}~\eqref{eq:CompactModel} and the \gls{a-pdgd} with input 
\eqref{eq:PdControllerWithInput} as follows:
\begin{subequations}\label{eq:Ic}
\begin{empheq}[left=\Sigma_{\mathrm I}:\;\empheqlbrace]{align}
u &= M_u \theta, \label{subeq:IcU}\\
v_{\mathrm{pd}} &= -\bigl(y - M_y \theta \bigr), \label{subeq:IcV}
\end{empheq}
\end{subequations}
where
\begin{subequations}\label{eq:M_u_and_M_y}
    \begin{align}
        M_u & = 2B_{\mathrm{pd}}^{\top} P_2=\begin{bmatrix}
            \mathbf{0}_{m\times n+p}  &  \mathbf{I}_{m\times m} &  \mathbf{0}_{m\times n_\mathrm{eq}+n_\mathrm{ineq}} 
        \end{bmatrix}, \label{eq:M_u}\\
    M_y & = \begin{bmatrix}
 \mathbf{0}_{p\times n} &  \mathbf{I}_{p\times p} &  \mathbf{0}_{p\times m+n_\mathrm{eq}+n_\mathrm{ineq}} 
    \end{bmatrix},
    \end{align}
\end{subequations}
with $B_{\mathrm{pd}}$ and $P_2$  defined in~\eqref{eq:PdControllerWithInput} and~\eqref{eq:matrix_P_2}, respectively. With the proposed choice of $M_u$, it is straightforward to verify that $M_u \theta = \hat{u}$. Then, \eqref{eq:Ic} implies that the input to the  \gls{ncs}~\eqref{eq:CompactModel} is assigned to be equal to the passive output of the Aug-PDGD.
Moreover, \(M_y\) is such that \(M_y\theta=\hat y\), where $\hat y$ is the Aug-PDGD estimate of the passive output \(y\) of the CPS dynamics~\eqref{eq:CompactModel}. Then, \eqref{eq:Ic} implies that $v_\mathrm{pd}$ vanishes at equilibrium, making the closed-loop system's equilibrium coincide with the solution of the \gls{kkt} conditions~\eqref{eq:KktConditions}. In essence, \eqref{eq:Ic} is a feedback, power-preserving interconnection between two (strictly) shifted-passive systems. Finally, note that $B_\mathrm{pd}$ can be computed as $B_\mathrm{pd}=\frac{1}{2}P_2^{-1}M_u^\top.$

\section{Conditions for Closed-Loop Stability} \label{sec:ClosedLoopStab}

In this section, we provide sufficient tuning conditions for the controller~\eqref{eq:PdControllerWithInput},~\eqref{eq:Ic} under which the overall closed-loop system, conformed by \eqref{eq:CompactModel}, \eqref{eq:PdControllerWithInput} and \eqref{eq:Ic},  possesses a unique globally exponentially stable equilibrium.

\begin{theorem}
Consider the closed-loop system $\Sigma_{\mathrm{cl}} = \Sigma_{\mathrm{p}} \circ \Sigma_{\mathrm{I}} \circ \Sigma_{\mathrm{c}}$.
Let

\begin{equation}\label{eq:value_beta}
  \beta \coloneqq \max_i\{\big(\lambda_i \big( M_u^\top M_y + M_y^\top M_u \big)\big)\}.
\end{equation}

Choose $\eta>0$ such that, with $\epsilon$ as defined in~\eqref{eq:ChoiceOfC1},
\begin{equation}\label{eq:choice_eta}
  \eta \kappa_1 \epsilon \min(\eta,1) > \beta.
\end{equation}
Then, the unique equilibrium point of the closed-loop system $\Sigma_{\mathrm{cl}}$ is exponentially stable with decay rate no slower than
\begin{subequations}\label{eq:tau_and_tau_e}
\begin{align}
     \tau  = \frac{1}{2} \min\big( \tau_1, \tau_{2,e} \big),
     \label{eq:tau}
\end{align}
where $0<\tau_1<2\min_i\{-\text{Re}(\lambda_i(A_\mathrm{p}))\}$ and, with $c$ from~\eqref{eq:ChoiceOfC_all}, 
\begin{equation}
         \tau_{2,e}  = \frac{2\tau_2 \, \epsilon c \, \min(\eta,1) - \beta}
                     {2\epsilon c \, \min(\eta,1)}.
\label{eq:tau_e}
\end{equation}
\end{subequations}
\hfill $\square$
\end{theorem}

\begin{proof}
    Consider the Lyapunov function \(S \coloneqq S_{\mathrm p}+S_{\mathrm c}\), where $S_\mathrm{p}$ and $S_\mathrm{c}$ are as in Lemmas~\ref{lemma:PlantPassivity} and \ref{lemma:PdPassivity}, respectively. Along any
trajectory of \(\Sigma_{\mathrm{cl}}\), $\dot S$ satisfies, in view of Lemmas~\ref{lemma:PlantPassivity} and \ref{lemma:PdPassivity},  
\begin{equation}
\dot{S}
\;\le\;
- \tau_1 \tilde{x}^\top P_1 \tilde{x}
+ \tilde{y}^\top \tilde{u}
- \tau_2 \tilde{\theta}^\top P_2 \tilde{\theta}
+ 2\,\tilde{\theta}^\top P_2 B_{\mathrm{pd}} \tilde{v}_{\mathrm{pd}} .
\label{eq:LS-start-short}
\end{equation}
Due to  the interconnection~\eqref{eq:Ic}, the inequality \eqref{eq:LS-start-short} simplifies to
\begin{align}
\dot{S}
&\le
- \tau_1 \tilde{x}^\top P_1 \tilde{x}
- \tau_2 \tilde{\theta}^\top P_2 \tilde{\theta}
+ \tilde{\theta}^\top M_u^\top M_y \tilde{\theta}.
\label{eq:cross-cancel-short}
\end{align}
Let \[
\beta = \max_i\{\big(\lambda_i \big( M_u^\top M_y + M_y^\top M_u \big)\big)\}.
\]
Then,
\begin{equation}
\tilde{\theta}^\top M_u^\top M_y \tilde{\theta}
\le \tfrac{\beta}{2}\|\tilde{\theta}\|^2.
\label{eq:betaBound-short}
\end{equation}
Due to \eqref{eq:lower_bound_P_2}, it holds that
\begin{align}
\|\tilde{\theta}\|^2
\le \frac{1}{\epsilon c \min(\eta,1)}\,\tilde{\theta}^\top P_2 \tilde{\theta}.  
\label{eq:thetaBoundbyP}
\end{align}
By combining \eqref{eq:thetaBoundbyP} with \eqref{eq:betaBound-short} we obtain that
\begin{align}
    \tilde{\theta}^\top M_u^\top M_y \tilde{\theta}
\le
\frac{\beta}{2\,\epsilon c \min(\eta,1)}\,
\tilde{\theta}^\top P_2 \tilde{\theta}.
\label{eq:BetaP_Bound}
\end{align}
Considering \eqref{eq:BetaP_Bound}, the following inequality is implied  from \eqref{eq:cross-cancel-short}:
\begin{equation}\label{eq:dot_s_second_to_last}
\dot{S}
\leq
- \tau_1 \tilde{x}^\top P_1 \tilde{x}
- \Big(\tau_2 - \frac{\beta}{2\,\epsilon c \min(\eta,1)}\Big)
\,\tilde{\theta}^\top P_2 \tilde{\theta}.
\end{equation}
From Lemma~\ref{lemma:PdPassivity}, \(\tau_2 = \frac{\eta\kappa_1}{2c}\). Then,
\[
\tau_2 - \frac{\beta}{2\,\epsilon c \min(\eta,1)}
= \frac{\eta\kappa_1}{2c} - \frac{\beta}{2\,\epsilon c \min(\eta,1)}.
\]
Note that if  \(\eta \ge 1\), then \(\min(\eta,1)=1\), and
\[
\frac{\eta\kappa_1}{2c} - \frac{\beta}{2\,\epsilon c}
= \frac{\eta\kappa_1\epsilon - \beta}{2\epsilon c}.
\]
The right-hand side of this equation is positive due to \eqref{eq:choice_eta}.
Alternatively, if  \(0<\eta<1\),  then \(\min(\eta,1)=\eta\), and
\[
\frac{\eta\kappa_1}{2c} - \frac{\beta}{2\,\epsilon c \eta}
= \frac{\eta^2\kappa_1\epsilon - \beta}{2\epsilon c}.
\]
Analogously, the right-hand side of this equation is positive due to \eqref{eq:choice_eta}. It follows that under the condition of the theorem \(\tau_2 - \frac{\beta}{2\,\epsilon c \min(\eta,1)} > 0\). Then,  choosing $\tau_{2,e}>0$ as in \eqref{eq:tau_e}  ensures  that
\[
\dot{S}
\le
- \tau_1 \tilde{x}^\top P_1 \tilde{x}
- \tau_{2,e} \tilde{\theta}^\top P_2 \tilde{\theta}.
\]
Consequently, the unique equilibrium of the closed-loop system $\Sigma_\mathrm{cl}$  is exponentially stable. Moreover, the decay rate is no slower than
\(\tau = \frac{1}{2} \min(\tau_1,\tau_{2,e})\) \cite[Theorem~4.10]{khalil2002nonlinear}.
\end{proof}

\begin{remark}
Assume that at each fault-induced switching instant the Lyapunov function $S$ may increase by at most a factor \(\gamma_{\mathrm f}\ge 1\), i.e., $S(t_k^{+}) \le \gamma_{\mathrm f}\, S(t_k^{-})$. Between switches, each mode (i.e., the closed-loop dynamics obtained by interconnecting \eqref{eq:CompactModel}, \eqref{eq:PdControllerWithInput}, and \eqref{eq:Ic} for a fixed fault configuration) has an exponentially stable equilibrium with decay rate
\(\tau>0\), i.e., \(\dot S \le -2\tau S\). Then, a sufficient dwell-time ensuring
overall exponential stability, even under consecutive faults, is given by $T_{\mathrm{dwell}} > \frac{\ln \gamma_{\mathrm f}}{2\tau}$, which follows from standard average dwell-time arguments for switched systems
with multiple Lyapunov functions \cite[Theorem~3.2]{liberzon2003}.
\hfill $\square$
\end{remark}

\section{Application to DC Microgrids}\label{sec:application_microgrids}

\begin{figure*}[t]
  \centering
  \begin{subfigure}[t]{0.49\textwidth}
    \centering
    \includegraphics[width=\linewidth]{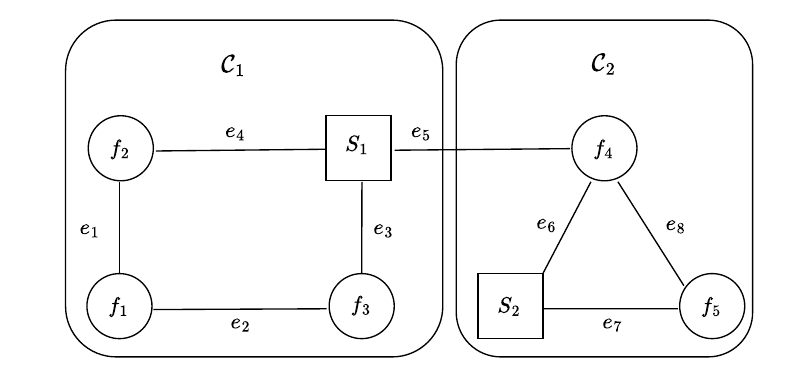}
    \caption{DC microgrid partitioned into two clusters $\mathcal{C}_1$ and $\mathcal{C}_2$. }
    \label{fig:Cluster}
  \end{subfigure}\hfill
  \begin{subfigure}[t]{0.49\textwidth}
    \centering
    \input{figs/fig_buses_tikz} 
    \caption{Schematic of voltage-following (left) and voltage-setting (right) buses.}
    \label{fig:Buses}
  \end{subfigure}
  \caption{System overview: (a) DC microgrid partitioned into clusters and (b) bus-type schematic.}
  \label{fig:ClusterAndBuses}
\end{figure*}

Driven by environmental goals, the growing penetration of renewables is accelerating the adoption of flexible distribution architectures such as microgrids. Microgrids comprise interconnected \gls{dgu}s, storage units, and loads, and can operate grid-connected or islanded~\cite{Cucuzzella}; here we focus on the islanded case. Microgrids remain vulnerable to faults that can degrade performance or compromise stability, motivating robust and \gls{ftc} strategies~\cite{Huang}. The  \gls{ftc} framework developed in this paper for general \glspl{ncs} is well suited to address these challenges, as illustrated next on a representative DC microgrid example.

\subsection{System Model}

Consider a network of \(n\ge1\) DC microgrid clusters
\(\mathcal{C}_1,\dots,\mathcal{C}_n\) interconnected via tie-lines; see Fig.~\ref{fig:Cluster} for an illustrative example with two clusters. We model each cluster $\mathcal{C}_i$ as an arbitrarily oriented graph $\mathcal{G}_i=(\mathcal{N}_i,\mathcal{E}_i)$, where each element in $\mathcal{N}_i$ represents an electric bus, and each element in $\mathcal{E}_i$, an inductive-resistive power line. We consider the following splitting of the elements in $\mathcal{N}_i$ and $\mathcal{E}_i$:
\begin{align*}
    \mathcal{N}_i & =\mathcal{N}_{\mathrm{s},i}^\mathrm{a}\cup \mathcal{N}_{\mathrm{f},i}^\mathrm{a}\cup \mathcal{N}_{\mathrm{s},i}^\mathrm{b}\cup \mathcal{N}_{\mathrm{f},i}^\mathrm{b}\\
 \mathcal{E}_i & =\mathcal{E}^\mathrm{a}_i\cup  \mathcal{E}_i^\mathrm{ab}.
\end{align*}
Each element  $j\in \mathcal{N}_{\mathrm{s},i}^\mathrm{a}$ represents an electric bus \emph{within} the cluster $\mathcal{C}_i$,  whose voltage $v_{\mathrm{s},j}$ is directly assigned by a voltage-setting DGU. We assume that each cluster has at least one voltage-setting bus, i.e., $\vert \mathcal{N}_{\mathrm{s},i}^\mathrm{a}\vert \geq 1$.
Each element $j\in \mathcal{N}_{\mathrm{f},i}^\mathrm{a}$ represents a capacitive bus, also within the cluster $\mathcal{C}_i$, with voltage $v_{\mathrm{f},j}$ and capacitance $c_{\mathrm{f},j}$, to which a voltage-following DGU directly injects a current $i_{\mathrm{f},j}$, and which energizes a ZI-load with admittance $\psi_{\mathrm{load},j}$ and constant-current term $\delta_{\mathrm{load},j}$. Analogous descriptions follow for the elements in $\mathcal{N}_{\mathrm{s},i}^\mathrm{b}$ and $\mathcal{N}_{\mathrm{f},i}^\mathrm{b}$, with the distinction that their elements are associated to  buses of any other cluster different from $\mathcal{C}_i$ which is physically linked to $\mathcal{C}_i$ via a tie-line.

Each tie-line is assigned to exactly one of its incident clusters. Tie-lines never connect two voltage-setting buses; if a tie-line connects a voltage-setting bus to a voltage-following bus, it is assigned to the cluster containing the voltage-setting bus. Accordingly, $\mathcal{E}_i^\mathrm{a}$ collects all lines owned by $\mathcal{C}_i$ (internal lines and assigned tie-lines), whereas $\mathcal{E}_i^\mathrm{ab}$ contains tie-lines incident to $\mathcal{C}_i$ but assigned to adjacent clusters. Accordingly, $ \mathcal{E}_i^\mathrm{a} \subset \bigl(\mathcal{N}_i^\mathrm{a} \times \mathcal{N}_i^\mathrm{a}\bigr) \cup \bigl(\mathcal{N}_i^\mathrm{a} \times \mathcal{N}_i^\mathrm{b}\bigr).$ Alternatively, each element in \(\mathcal{E}_i^\mathrm{ab}\) represents a tie-line linking a bus within cluster \(\mathcal{C}_i\) to a bus \(k \in \mathcal{N}_i^\mathrm{b}\) belonging to an adjacent cluster, \emph{and assigned to that adjacent cluster}, i.e., $\mathcal{E}_i^\mathrm{ab} \subset \mathcal{N}_i^\mathrm{a} \times \mathcal{N}_i^\mathrm{b}.$ The current and inductance--resistance pair of such lines adopts the same notation as for the elements in \(\mathcal{E}_i^\mathrm{a}\).  

For each cluster $\mathcal{C}_i$, let us introduce the vectors $V_{\mathrm{s,a},i}=\text{col}(\{v_{\mathrm{s},j}\}_{j\in \mathcal{N}_{\mathrm{s},i}^\mathrm{a}})$ and $V_{\mathrm{f,a},i}=\text{col}(\{v_{\mathrm{f},j}\}_{j\in \mathcal{N}_{\mathrm{f},i}^\mathrm{a}})$, with analogous definitions holding for $V_{\mathrm{s,b},i}$ and $V_{\mathrm{f,b},i}$. Moreover, let $I_{\mathrm{a},i}=\text{col}(\{i_{\mathrm{line},j}\}_{j\in \mathcal{E}_i^\mathrm{a}})$ and $I_{\mathrm{ab},i}=\text{col}(\{i_{\mathrm{line},j}\}_{j\in \mathcal{E}_i^\mathrm{ab}})$. By introducing the node-edge incidence matrix $\mathcal{B}_i$ for cluster $\mathcal{C}_i$ as 
\begin{equation*}
\mathcal{B}_i^\top
=
\begin{bmatrix}
\left(\mathcal{B}_{\mathrm{s,a},i}^\mathrm{a}\right)^\top &
\left(\mathcal{B}_{\mathrm{f,a},i}^\mathrm{a}\right)^\top &
\left(\mathcal{B}_{\mathrm{s,b},i}^\mathrm{a}\right)^\top &
\left(\mathcal{B}_{\mathrm{f,b},i}^\mathrm{a}\right)^\top\\
\left(\mathcal{B}_{\mathrm{s,a},i}^\mathrm{ab}\right)^\top &
\left(\mathcal{B}_{\mathrm{f,a},i}^\mathrm{ab}\right)^\top &
\left(\mathcal{B}_{\mathrm{s,b},i}^\mathrm{ab}\right)^\top &
\left(\mathcal{B}_{\mathrm{f,b},i}^\mathrm{ab}\right)^\top
\end{bmatrix},
\end{equation*}
then it is possible to obtain the following dynamical model for $\mathcal{C}_i$:
\begin{center}
\scalebox{1}{
\begin{minipage}{\linewidth}
\begin{subequations}\label{eq:MgClusterModel-short}
\begin{align}
C_{\mathrm{f,a},i}\dot V_{\mathrm{f,a},i} &=
- Y_{\mathrm{load},i}V_{\mathrm{f,a},i} - I_{\mathrm{load},i}+I_{\mathrm{f,a},i} - \mathcal{B}_{\mathrm{f,a},i}^\mathrm{a} I_{\mathrm{a},i}- \mathcal{B}^{\mathrm{ab}}_{\mathrm{f,a},i} I_{\mathrm{ab},i}
,\\
L_{\mathrm{a},i}\dot I_{\mathrm{a},i} &=
- R_{\mathrm{a},i} I_{\mathrm{a},i}
+ (\mathcal{B}_{\mathrm{s,a},i}^\mathrm{a})^\top V_{\mathrm{s,a},i}
+ (\mathcal{B}_{\mathrm{f,a},i}^\mathrm{a})^\top V_{\mathrm{f,a},i} 
+ (\mathcal{B}_{\mathrm{f,b},i}^\mathrm{a})^\top V_{\mathrm{f,b},i},
\end{align}
\end{subequations}
\end{minipage}
}
\end{center}
where $C_{\mathrm{f,a},i}=\text{diag}(\{c_{\mathrm{f},j}\}_{j\in \mathcal{N}_{\mathrm{f},i}^\mathrm{a}})$, $Y_{\mathrm{load},i}=\text{diag}(\{\psi_{\mathrm{load} ,j}\}_{j\in \mathcal{N}_{\mathrm{f},i}^\mathrm{a}})$, $I_{\mathrm{load},i}=\text{col}(\{\delta_{\mathrm{load},j}\}_{j\in \mathcal{N}_{\mathrm{f},i}^\mathrm{a}})$ and $I_{\mathrm{f,a},i}=\text{col}(\{i_{\mathrm{f},j}\}_{j\in \mathcal{N}_{\mathrm{f},i}^\mathrm{a}})$; moreover, $L_{\mathrm{a},i}=\text{diag}(\{\ell_{\mathrm{line},j}\}_{j\in \mathcal{E}_i^\mathrm{a}})$ and $R_{\mathrm{a},i}=\text{diag}(\{r_{\mathrm{line},j}\}_{j\in \mathcal{E}_i^\mathrm{a}})$. 
The model \eqref{eq:MgClusterModel-short} can be written in the  state space form~\eqref{eq:plant_i} by defining the state vector $x_i =\mathrm{col}(V_{\mathrm{f,a},i},I_{\mathrm{a},i})$, the measured output $y_i = \mathrm{col}(V_{\mathrm{f,a},i},I_{\mathrm{a},i})$, the interconnection output $z_i = \mathrm{col}(\mathcal{B}_{\mathrm{f,a},i}^\mathrm{ab} V_{\mathrm{f,a},i}, \mathcal{B}_{\mathrm{f,b},i}^{\mathrm{a}} I_{\mathrm{a},i})$, the vector of control inputs $u_i = \mathrm{col}(I_{\mathrm{f,a},i},V_{\mathrm{s,a},i})$, the vector of external variables $ w_i = \mathrm{col}(I_{\mathrm{ab},i}, V_{\mathrm{f,b},i})$, with $(A_i,B_i,d_i,G_i,E_i,C_i)$ defined accordingly. Moreover, the admissible steady-state sets for each cluster $\mathcal{C}_i$ are
$\mathcal{X}_i \coloneqq \{\,x_i \mid V_{\mathrm{f,a},i} \ge V_{\mathrm{f,a},i}^{\min}\,\}$,
$\mathcal{U}_i \coloneqq \{\,u_i \mid I_{\mathrm{f,a},i} \le I_{\mathrm{f,a},i}^{\max},~ V_{\mathrm{s,a},i} \ge V_{\mathrm{s,a},i}^{\min}\,\}$, and
$\mathcal{Y}_i \coloneqq \{\,y_i \mid y_i=x_i,~ V_{\mathrm{f,a},i} \ge V_{\mathrm{f,a},i}^{\min}\,\}$ \footnote{From a practical standpoint, we impose lower bounds on the setting-bus voltages to avoid undervoltage and upper bounds on the follower-bus injected currents to respect source current ratings, which may become tighter after a fault.}.

The interconnection between clusters $\mathcal{C}_i$ and $\mathcal{C}_j$ is
\[
\Omega_{ij}\coloneqq
\begin{bmatrix}
\mathbf{0} & -P^{\mathrm i}_{ij}\\
P^{\mathrm v}_{ij} & \mathbf{0}
\end{bmatrix},
\]
where $P^{\mathrm v}_{ij}$ and $P^{\mathrm i}_{ij}$ are Boolean selection matrices encoding the tie-line adjacency.
\footnote{%
$P^{\mathrm v}_{ij}$ selects the boundary-bus voltages of $\mathcal{C}_j$ that appear as remote voltages in $V_{\mathrm{f,b},i}$,
and $P^{\mathrm i}_{ij}$ selects the tie-line currents owned by $\mathcal{C}_j$ that appear as incident (non-owned) currents in $I_{\mathrm{ab},i}$.
The minus sign enforces opposite current-flow conventions at the two ends of each tie-line.}
By consistent stacking of shared tie-lines, the selection matrices satisfy
$P^{\mathrm i}_{ij}=(P^{\mathrm v}_{ji})^\top$ (equivalently $P^{\mathrm v}_{ij}=(P^{\mathrm i}_{ji})^\top$), hence
$\Omega_{ij}+\Omega_{ji}^\top=\mathbf{0}$ and Assumption~\ref{assump:interconnection_passive} holds.

\subsection{Fault Scenario and Numerical Simulations}

\begin{figure*}[t]
    \centering

    % --- Left: merged voltages ---
    \begin{subfigure}[t]{0.48\textwidth}
        \centering
        \includegraphics[width=\linewidth]{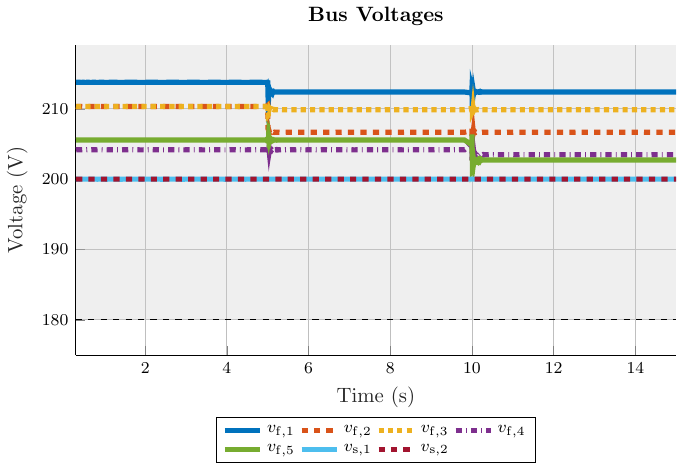}
        \caption{Follower- and setting-bus voltages $v_{\mathrm{f},j}$ and $v_{\mathrm{s},k}$.}
        \label{fig:sub_vfvs}
    \end{subfigure}
    \hfill
    % --- Right: injected currents ---
    \begin{subfigure}[t]{0.48\textwidth}
        \centering
        \includegraphics[width=\linewidth]{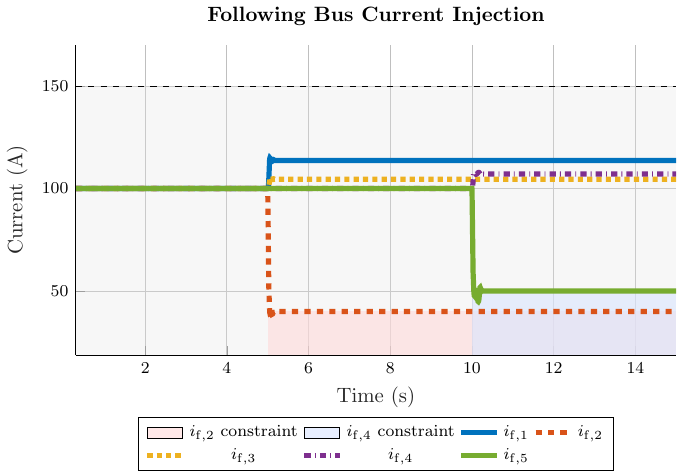}
        \caption{Injected currents $i_{\mathrm{f},j}$ under changing limits.}
        \label{fig:sub_if}
    \end{subfigure}

    \caption{Closed-loop trajectories of the DC microgrid under the proposed FTC scheme.}
    \label{fig:ftc_dcmg_timeseries}
\end{figure*}

We focus on faults that manifest as a power loss at a voltage-following bus. Then, for some voltage-following DGU $j\in \mathcal{N}^{\mathrm{a}}_{\mathrm{f},i}$ within a cluster $\mathcal{C}_i$, a step reduction in the current injection $i_{\mathrm{f},j}$ occurs at time $t_\mathrm{f}$.  Consequently, the objective of the FTC \eqref{eq:PdControllerWithInput},~\eqref{eq:Ic}, applied to the DC microgrid model,  is to drive the system towards a new equilibrium whose error with respect to the pre-fault  pre-fault equilibrium is minimum, while respecting all physical and fault-induced constraints. The desired post-fault steady-state is a solution to an optimization problem of the form \eqref{eq:FTCObjectives3}, with cost function
\begin{equation*}
    J(\bar \xi)=\Vert \bar \xi-\xi^*\Vert_K^2,
\end{equation*}
with $\xi^*=\mathrm{col}(x^*,u^*)$, 
$\bar{\xi}=\mathrm{col}(\bar x, \bar u)$,  
$
K=\mathrm{blkdiag}(K_x,K_u)\succ \mathbf{0}$,
and steady-state equality and inequality constraints given by $R_\mathrm{eq}\bar \xi = b$ and $R_\mathrm{ineq}\bar \xi \leq h$, with

\[
R_{\mathrm{eq}}=
\begin{bmatrix}
- Y_{\mathrm{load}} & - \mathcal B_{\mathrm{f}} & \mathbf{I} & \mathbf{0}\\
(B_{\mathrm{f}})^\top & - R_\mathrm{a} & \mathbf{0} & \mathcal (B_{\mathrm{s}})^\top
\end{bmatrix},
\quad
b=\begin{bmatrix} I_{\mathrm{load}}  \\ \mathbf{0} \end{bmatrix},
\]

\[
R_{\mathrm{ineq}}=
\begin{bmatrix} \mathbf{0} & \mathbf{0} & \mathbf{I} & \mathbf{0} \\
\mathbf{0} & \mathbf{0} & \mathbf{0} & -\mathbf{I}  \end{bmatrix},
\quad
h=\begin{bmatrix}  I^{\max}_{\mathrm{f}}  \\ - V^{\min}_{\mathrm{s}}   \end{bmatrix},
\]
where $Y_{\mathrm{load}}\coloneqq \mathrm{blkdiag}(Y_{\mathrm{load},1},\dots,Y_{\mathrm{load},n}),\ \ R_{\mathrm a}\coloneqq \mathrm{blkdiag}(R_{\mathrm a,1},\dots,R_{\mathrm a,n}),\ \ I_{\mathrm{load}}\coloneqq \mathrm{col}(I_{\mathrm{load},1},\dots,I_{\mathrm{load},n}).$ Moreover, $\mathcal{B}_{\mathrm{f}}$ and $\mathcal{B}_{\mathrm{s}}$ denote the global node--edge incidence matrices corresponding to the voltage-following and voltage-setting buses, respectively.

For the numerical experiments, we consider the DC microgrid shown in Fig.~\ref{fig:Cluster}.  All parameters defining the system model and the steady-state operational constraints are provided in the accompanying online repository~\cite{zaidi2025dcmg}\footnote{Using these values, it can be verified that the assumptions and design conditions stated in this paper are satisfied for the considered system.}. Figure~\ref{fig:ftc_dcmg_timeseries} shows the trajectories of  \(i_{\mathrm f}\), and the bus voltages \(v_{\mathrm f},v_{\mathrm s}\).
At \(t=5\,\mathrm{s}\), the voltage-following bus~2 loses power capacity, reducing its injection limit from \(150\,\mathrm{A}\) to \(40\,\mathrm{A}\); at \(t=10\,\mathrm{s}\), the voltage-following bus~5 is reduced to \(50\,\mathrm{A}\).
Despite these faults, the proposed FTC keeps voltages and currents within bounds and converges to a feasible operating point close to the pre-fault one.

\section{Conclusions} \label{sec:Conclusion}
In this paper, we have presented an active primal-dual FTC scheme based on an \gls{a-pdgd} for the post-fault exponential stabilization of a class of networked CPS. The proposed approach---which is framed as a CbI control scheme---was designed to provably achieve convergence towards the unique (KKT) solution of a constrained optimization problem encoding desired post-fault operation. This addresses the lack, highlighted in the literature review, of network-level, constraint-aware optimal coordination with explicit optimality guarantees under shared-resource and interconnection constraints in \gls{ncs} such as power and energy networks (e.g., microgrids). Numerical simulations on a representative DC microgrid showed that the proposed method ensures post-fault closed-loop stability and maintains pre-specified steady-state operational constraints (such as bus-voltage and line–current limits) in the presence of non-trivial faults, such as loss of power capacity at the generation units. Future work will focus on providing conditions for enabling a distributed implementation of the proposed controller and on its real-time experimental validation in a power hardware-in-the-loop environment.

\section*{Acknowledgment}
\noindent This work builds on earlier unpublished work conducted while the second author was at University of Groningen (2020-2023). We acknowledge S. Ahmed, J. Ferguson, M. Cucuzzella and J.M.A. Scherpen for their contributions to those efforts.

\bibliographystyle{IEEEtran}
\bibliography{references}

\appendix
\subsection{Proof of Lemma \ref{lemma:PdPassivity}}  \label{ProofLemma:PdPassivity}

To prove Lemma~2, we follow a proof strategy similar to the inequality-only case in \cite{qu_li_augPD}. To improve readability, key facts relevant for proving  Lemma~2 are presented in the form of preliminary propositions, after which we give  a re-statement of Lemma~2  and  present its proof.

\begin{proposition}\label{prop:lower_bound_P2}
Let $P_2$ be defined in~\eqref{eq:matrix_P_2}. Let $\epsilon\in(0,1)$ and let
$R$ be the stacked constraint matrix defined in~\eqref{eq:R_stacked}. Suppose that
$RR^\top \preceq \kappa_2 \mathbf I$ holds (Assumption~\ref{assump:BoundOnR}).
If $c$ is such that \eqref{eq:ChoiceOfC1} holds, i.e., 
\begin{equation*}
    c \ge
\begin{cases}
\displaystyle \eta\,\sqrt{\frac{\kappa_{2}}{(\eta - \epsilon)(1-\epsilon)}}, 
& \eta \ge 1,\\[3mm]
\displaystyle \sqrt{\frac{\eta\,\kappa_{2}}{(1 - \epsilon)(1-\eta \epsilon)}}, 
& 0 < \eta < 1,
\end{cases}
\end{equation*}
then
\begin{equation}\label{eq:lower_bound_P_2}
P_2 \;\succeq\; \epsilon\,c\,\min(\eta,1)\,\mathbf I,
\end{equation}
and, in particular, $P_2$ is positive definite.
\hfill $\square$
\end{proposition}

\begin{proof}
Consider
\[
\widetilde P_2 \;\coloneqq\; P_2 - \epsilon c \min(\eta,1)\,\mathbf I.
\]
Using the stacked matrix $R$, the matrix $P_2$ in~\eqref{eq:matrix_P_2} can be written as
\[
P_2 \;=\;
\begin{bmatrix}
\eta c\,\mathbf I & \eta R^\top\\
\eta R & c\,\mathbf I
\end{bmatrix},
\]
and hence
\[
\widetilde P_2 \;=\;
\begin{bmatrix}
c\bigl(\eta-\epsilon\min(\eta,1)\bigr)\mathbf I & \eta R^\top\\
\eta R & c\bigl(1-\epsilon\min(\eta,1)\bigr)\mathbf I
\end{bmatrix}.
\]
We show $\widetilde P_2\succeq 0$ by a Schur-complement argument.

\medskip
\noindent\emph{Case 1: $\eta\ge 1$.}
Then $\min(\eta,1)=1$ and
\[
\widetilde P_2 \;=\;
\begin{bmatrix}
c(\eta-\epsilon)\mathbf I & \eta R^\top\\
\eta R & c(1-\epsilon)\mathbf I
\end{bmatrix}.
\]
Since $\eta>\epsilon$, the leading block $c(\eta-\epsilon)\mathbf I \succ 0$.
By the Schur complement, $\widetilde P_2\succeq 0$ is implied by
\[
c(1-\epsilon)\mathbf I \;-\; \frac{\eta^2}{c(\eta-\epsilon)}\,RR^\top \;\succeq\; 0.
\]
Using $RR^\top \preceq \kappa_2 \mathbf I$, it is sufficient that
\begin{align}\label{eq:ineqCase1_P2}
    c(1-\epsilon) \;-\; \frac{\eta^2}{c(\eta-\epsilon)}\,\kappa_2 \;& \ge\; 0 \nonumber\\
 \quad\Longleftrightarrow\quad \nonumber\\
c^2(1-\epsilon)(\eta-\epsilon) \; &\ge\; \eta^2\kappa_2.
\end{align}
\medskip
\noindent\emph{Case 2: $0<\eta<1$.}
Then $\min(\eta,1)=\eta$ and
\[
\widetilde P_2 \;=\;
\begin{bmatrix}
\eta c(1-\epsilon)\mathbf I & \eta R^\top\\
\eta R & c(1-\eta\epsilon)\mathbf I
\end{bmatrix}.
\]
Since $\epsilon\in(0,1)$, the leading block $\eta c(1-\epsilon)\mathbf I \succ 0$.
By the Schur complement, $\widetilde P_2\succeq 0$ is implied by
\[
c(1-\eta\epsilon)\mathbf I \;-\; \frac{\eta}{c(1-\epsilon)}\,RR^\top \;\succeq\; 0.
\]
Using $RR^\top \preceq \kappa_2 \mathbf I$, it is sufficient that
\begin{align}\label{eq:ineqCase2_P2}
    c(1-\eta\epsilon) \;-\; \frac{\eta}{c(1-\epsilon)}\,\kappa_2 \; & \ge \; 0 \nonumber\\
 & \quad\Longleftrightarrow\quad \nonumber\\
 c^2(1-\epsilon)(1-\eta\epsilon) & \ge\; \eta\kappa_2.
\end{align}
\medskip
Then,  if  $c$ is as in the Proposition's statement, we can see directly that
\eqref{eq:ineqCase1_P2} holds in Case~1 and \eqref{eq:ineqCase2_P2} holds in Case~2, and hence $\widetilde P_2\succeq 0$. Therefore, $P_2 \succeq \epsilon c \min(\eta,1)\mathbf I$, which proves~\eqref{eq:lower_bound_P_2}.
Since $\epsilon c \min(\eta,1)>0$, this further implies $P_2\succ 0$.
\end{proof}

\begin{proposition}\label{prop:shifted_primal_dual}
    Let
    \begin{equation}\label{eq:phi_i}
    \phi_i:= \nu_{\mathrm{ineq},i} + \rho\,(R_{\mathrm{ineq}}\xi - h)_i.
\end{equation}
    Then, there exists a symmetric matrix $H(\xi)\in \mathbb{R}^{n_\xi\times n_\xi}$ and functions $\phi_i\mapsto \gamma_i(\phi_i)$ 
    satisfying, for $i=1,2,...,n_\mathrm{ineq}$,
\begin{equation}\label{eq:properties_gamma_i}
    \gamma_i(\phi_i)\in [0,1],   
    \end{equation}
 such that the shifted dynamics (error dynamics) of the Aug-PDGD \eqref{eq:PdControllerWithInput} can be written as
\begin{equation}\label{eq:shiftedPd-short}
\dot{\tilde{\theta}}
= F(\theta)\tilde{\theta} + B_{\mathrm{pd}}\tilde v_{\mathrm{pd}},
\end{equation}
with  
\begin{subequations}
\begin{equation}\label{eq:F_theta}
F(\theta)=
\begin{bmatrix}
- H(\xi) - \rho R_{\mathrm{ineq}}^\top \Gamma(\phi) R_{\mathrm{ineq}}
&
- R_{\mathrm{eq}}^\top
&
- R_{\mathrm{ineq}}^\top \Gamma (\phi)
\\[2mm]
\eta R_{\mathrm{eq}} & \mathbf{0} & \mathbf{0}
\\[2mm]
\eta \Gamma(\phi) R_{\mathrm{ineq}} & \mathbf{0} &
\frac{\eta}{\rho}(\Gamma(\phi) - \mathbf{I})
\end{bmatrix}
\end{equation}
and
\begin{equation}\label{eq:diag_mat_GAMMA}
\Gamma(\phi):=\operatorname{diag}(\gamma_1(\phi_1),\ldots,\gamma_{n_{\mathrm{ineq}}}(\phi_{n_{\mathrm{ineq}}})).
\end{equation}
\end{subequations}
Moreover, the matrix \(H(\xi)\) satisfies  
\begin{equation}\label{eq:bounds_H_xi}
    \mu\,\mathbf{I} \le H(\xi) \le \ell\,\mathbf{I},
\end{equation}
where $0<\mu\leq \ell$ are as in Assumption~\ref{assump:uConvex}.
\hfill $\square$
\end{proposition}
\begin{proof}
    Recall that $\bar{\theta} = (\bar{\xi},\bar{\nu}_{\mathrm{eq}},\bar{\nu}_{\mathrm{ineq}})$ is the unique solution of the KKT conditions \eqref{eq:KktConditions}, and is at the same time the unique equilibrium of the \gls{a-pdgd} \eqref{eq:PdControllerWithInput} when $v_\mathrm{pd}=\bar{v}_\mathrm{pd}=\boldsymbol{0}$. By defining  $\tilde{\theta}=\theta-\bar{\theta}$ and $\tilde{v}_\mathrm{pd}=v_\mathrm{pd}-\bar{v}_\mathrm{pd}$, it holds that
\begin{equation}\label{eq:shifted_augPD}
\dot{\tilde{\theta}}=
\begin{pmatrix}
- \big(\nabla J(\xi)-\nabla J(\bar{\xi})\big)
- R_{\mathrm{eq}}^\top \tilde{\nu}_{\mathrm{eq}}
- R_{\mathrm{ineq}}^\top
\big(
g(\xi,\nu_{\mathrm{ineq}})
-
g(\bar{\xi},\bar{\nu}_{\mathrm{ineq}})
\big)
\\[1mm]
\eta R_{\mathrm{eq}}\tilde{\xi}
\\[1mm]
\frac{\eta}{\rho}
\Big(
g(\xi,\nu_{\mathrm{ineq}})
-
g(\bar{\xi},\bar{\nu}_{\mathrm{ineq}})
-
\tilde{\nu}_{\mathrm{ineq}}
\Big)
\end{pmatrix} + B_{\mathrm{pd}} \tilde{v}_{\mathrm{pd}},
\end{equation}

By Assumption~\ref{assump:uConvex} and the mean value theorem, then \cite{qu_li_augPD2} 
\begin{align}
\nabla J(\xi)-\nabla J(\bar{\xi}) = H(\xi)(\xi- \bar{\xi}),
\label{eq:JtoH}
\end{align}
where
\[
H(\xi) := \int_0^1 \nabla^2 J\big(\bar{\xi} + s(\xi-\bar{\xi})\big)\, ds.
\]
Note already that due to Assumption~\ref{assump:uConvex},   \(\mu \mathbf{I} \preceq H(\xi) \preceq \ell \mathbf{I}\) for all \(\xi\), which establishes \eqref{eq:bounds_H_xi}.

Moving on, we focus now on finding an alternative representation for the term $g(\xi,\nu_{\mathrm{ineq}})
-
g(\bar{\xi},\bar \nu_\mathrm{ineq})$ in \eqref{eq:shifted_augPD}. For that purpose, let us introduce $\phi_i$ as in \eqref{eq:phi_i}.
Then, $g(\xi,\nu_{\text{ineq}})
\;=\;
\max\!\Big(
\nu_{\text{ineq}} + \rho(R_{\text{ineq}}\xi - h),\; 0
\Big),$ from \eqref{eq:vector_g}, can be written component-wise as
\begin{align}
g(\xi,\nu_{\text{ineq}})_i \;\coloneqq\; \max(\phi_i,0). \label{eq:vector_g_comp}
\end{align}
For brevity, we suppress in the sequel the dependence of $g_i$ on $(\xi,\nu_{\text{ineq}})$. Likewise, define
$\bar\phi_i \coloneqq \bar\nu_{\mathrm{ineq},i} + \rho\,(R_{\mathrm{ineq}}\bar\xi - h)_i$ so that
$\bar g_i = \max(\bar\phi_i,0)$.
Then the following identity holds:
\[
\max(\phi_i,0)-\max(\bar\phi_i,0)=\gamma_i(\phi_i)\,(\phi_i-\bar\phi_i),
\]
where $\gamma_i$, defined as,
\begin{equation}\label{eq:gamma_i}
\gamma_i(\phi_i)\;\coloneqq\;
\begin{cases}
1, & \phi_i>0,\ \bar{\phi}_i>0,\\[1mm]
0, & \phi_i\le 0,\ \bar{\phi}_i\le 0,\\[2mm]
\dfrac{\phi_i}{\phi_i-\bar{\phi}_i}, & \phi_i>0,\ \bar{\phi}_i\le 0,\\[3mm]
\dfrac{-\bar{\phi}_i}{\phi_i-\bar{\phi}_i}, & \phi_i\le 0,\ \bar{\phi}_i>0,
\end{cases}
\end{equation}
can be directly shown to satisfy $\gamma_i(\phi_i)\in [0,1]$. Moreover, by noting that
\[
\phi_i-\bar\phi_i
=
\tilde{\nu}_{\mathrm{ineq},i}+\rho\,(R_{\mathrm{ineq}}\tilde{\xi})_i,
\]
we obtain that
\begin{align}
g_i-\bar g_i
=
\gamma_i(\phi_i)\Big(\tilde{\nu}_{\mathrm{ineq},i}
+\rho\,(R_{\mathrm{ineq}}\tilde{\xi})_i\Big).
\end{align}
Now define the diagonal matrix
\[
\Gamma(\phi):=\operatorname{diag}(\gamma_1(\phi_1),\ldots,\gamma_{n_{\mathrm{ineq}}}(\phi_{n_{\mathrm{ineq}}})).
\]
Then,
\begin{align}
    g - \bar{g} =  \Gamma(\phi) \tilde{\nu}_{\mathrm{ineq}} + \rho \Gamma(\phi) R_{\mathrm{ineq}} \tilde{\xi}
    \label{eq:gTilde}
\end{align}
Using \eqref{eq:JtoH} and \eqref{eq:gTilde} we can rewrite \eqref{eq:shifted_augPD} as 
\begin{equation}\label{eq:shifted_augPD_withF}
\dot{\tilde{\theta}}=
\begin{bmatrix}
- H(\xi) - \rho R_{\mathrm{ineq}}^\top \Gamma(\phi) R_{\mathrm{ineq}}
&
- R_{\mathrm{eq}}^\top
&
- R_{\mathrm{ineq}}^\top \Gamma (\phi)
\\[2mm]
\eta R_{\mathrm{eq}} & \mathbf{0} & \mathbf{0}
\\[2mm]
\eta \Gamma(\phi) R_{\mathrm{ineq}} & \mathbf{0} &
\frac{\eta}{\rho}(\Gamma(\phi) - \mathbf{I})
\end{bmatrix} \begin{bmatrix}
    \tilde{\xi} \\ \tilde{\nu}_{\mathrm{eq}} \\ \tilde{\nu}_{\mathrm{ineq}}
\end{bmatrix} + B_{\mathrm{pd}} \tilde{v}_{\mathrm{pd}}.
\end{equation}
\end{proof}

\begin{proposition}\label{prop:explicity_form_Q}
    Recall the matrix $P_2 =
\begin{bmatrix}
\eta c\, \mathbf{I} & \eta R_{\mathrm{eq}}^{\top} & \eta R_{\mathrm{ineq}}^{\top} \\
\eta R_{\mathrm{eq}} & c\, \mathbf{I} & \mathbf{0}\\
\eta R_{\mathrm{ineq}} & \mathbf{0} & c\, \mathbf{I}
\end{bmatrix}
$ from \eqref{eq:matrix_P_2}, the matrix $F(\theta)$ from \eqref{eq:F_theta}, and $\tau_2=\frac{\eta\kappa_1}{2c}$ from \eqref{eq:tau2_def}. Then, 
\begin{align}\label{eq:explicit_form_Q}
  Q & := -F(\theta)^\top P_2-P_2 F(\theta)-\tau_2 P_2 \nonumber\\
  & = \begin{bmatrix}
       2\eta c\,H
+ 2\eta c\rho R^\top \Gamma_2 R
- 2\eta^2 R^\top \Gamma_1 R
- \frac{\eta^2\kappa_1}{2}\,\mathbf{I} & \eta\Bigl(
{
H +
{\rho R^\top \Gamma_2 R - \frac{\eta\kappa_1}{2c}\mathbf{I}}
}
\Bigr) R^\top
+
\eta{\frac{\eta}{\rho}R^\top(\mathbf{I}-\Gamma_1)}\\
*^\top & \eta(\Gamma_1 R R^\top + R R^\top\Gamma_1)
+ \frac{2\eta c}{\rho}(\mathbf{I}-\Gamma_1)
- \frac{\eta\kappa_1}{2}\,\mathbf{I}
  \end{bmatrix}.
\end{align}

\hfill $\square$
\end{proposition}

\begin{proof}
We begin by writing $F(\theta)$ as a 2-by-2 block matrix. For that purpose, let us define
\begin{equation}\label{eq:Gamma_1and_Gamma_2}
\Gamma_1(\phi)=\operatorname{blkdiag}(\mathbf{I},\Gamma(\phi)),
\qquad 
\Gamma_2(\phi)=\operatorname{blkdiag}(\mathbf{0},\Gamma(\phi)), 
\end{equation}
where we recall that $\Gamma$ is the diagonal function matrix defined in \eqref{eq:diag_mat_GAMMA}, with $\phi$ defined by components in \eqref{eq:phi_i}.
Then,
\[
F(\theta)=
\begin{bmatrix}
-\,H\;-\;\rho\,R^{\top}\Gamma_{2}\,R
&
-\,R^{\top}\Gamma_{1}
\\[2mm]
\eta\,\Gamma_{1}R
&
\dfrac{\eta}{\rho}\big(\Gamma_{1}-\mathbf{I}\big)
\end{bmatrix},\]
where $R$ is defined as in \eqref{eq:R_stacked}.
Similarly, for the (symmetric, constant) matrix $P_2$ it is possible to write the following representation:
\[
P_2 =
\begin{bmatrix}
\eta c\, \mathbf{I} & \hspace{3mm} & \eta R^{\top} \\[0mm]
\hspace{0mm}        &              & \hspace{0mm}   \\[0mm]
\eta R              & \hspace{3mm} & c\, \mathbf{I}
\end{bmatrix}.
\]
It follows that
\begin{align}
P_2F(\theta) & = \begin{bmatrix}
\eta c\, \mathbf{I} & \hspace{3mm} & \eta R^{\top} \\[0mm]
\hspace{0mm}        &              & \hspace{0mm}   \\[0mm]
\eta R              & \hspace{3mm} & c\, \mathbf{I}
\end{bmatrix} \begin{bmatrix}
-\,H\;-\;\rho\,R^{\top}\Gamma_{2}\,R
&
-\,R^{\top}\Gamma_{1}
\\[2mm]
\eta\,\Gamma_{1}R
&
\dfrac{\eta}{\rho}\big(\Gamma_{1}-\mathbf{I}\big)
\end{bmatrix} \nonumber \\
& = 
\begin{bmatrix}
-\eta c\,H-\eta c\rho\,R^\top\Gamma_2 R+\eta^2\,R^\top\Gamma_1 R
& \hspace{3mm} &
-\eta c\,R^\top\Gamma_1+\dfrac{\eta^2}{\rho}\,R^\top(\Gamma_1-\mathbf I)
\\[0mm]
\hspace{0mm} & & \hspace{0mm}
\\[0mm]
-\eta RH-\eta\rho\,R R^\top\Gamma_2 R+c\eta\,\Gamma_1 R
& \hspace{3mm} &
-\eta\,R R^\top\Gamma_1+\dfrac{c\eta}{\rho}(\Gamma_1-\mathbf I)
\end{bmatrix}
\end{align}
and
\begin{align}
F^\top (\theta) P_2 & =
\begin{bmatrix}
    -H^\top-\rho R^\top \Gamma_2R & \eta R^\top \Gamma_1\\
    -\Gamma_1 R & \frac{\eta}{\rho}(\Gamma_1-\mathbf{I})
\end{bmatrix} \begin{bmatrix}
\eta c\, \mathbf{I} & \hspace{3mm} & \eta R^{\top} \\[0mm]
\hspace{0mm}        &              & \hspace{0mm}   \\[0mm]
\eta R              & \hspace{3mm} & c\, \mathbf{I}
\end{bmatrix} \nonumber\\
&  = 
\begin{bmatrix}
-\eta c\,H-\eta c\rho\,R^\top\Gamma_2 R
+\eta^2\,R^\top\Gamma_1 R
& \hspace{3mm} &
-\eta H R^\top
-\eta\rho\,R^\top\Gamma_2 R R^\top
+c\eta\,R^\top\Gamma_1\\[0mm]
\hspace{0mm} & & \hspace{0mm}\\[0mm]
-\eta c\,\Gamma_1 R
+\frac{\eta^2}{\rho}\,(\Gamma_1-\mathbf I)R
& \hspace{3mm} &
-\eta\,\Gamma_1 R R^\top
+\frac{c\eta}{\rho}(\Gamma_1-\mathbf I)
\end{bmatrix}
\end{align}
Consequently,
\begin{align}
    Q & =-F^\top(\theta) P_2 - P_2F(\theta) - \tau_2 P_2 \nonumber\\
    & = -\begin{bmatrix}
-\eta c\,H-\eta c\rho\,R^\top\Gamma_2 R
+\eta^2\,R^\top\Gamma_1 R
& \hspace{3mm} &
-\eta H R^\top
-\eta\rho\,R^\top\Gamma_2 R R^\top
+c\eta\,R^\top\Gamma_1\\[0mm]
\hspace{0mm} & & \hspace{0mm}\\[0mm]
-\eta c\,\Gamma_1 R
+\frac{\eta^2}{\rho}\,(\Gamma_1-\mathbf I)R
& \hspace{3mm} &
-\eta\,\Gamma_1 R R^\top
+\frac{c\eta}{\rho}(\Gamma_1-\mathbf I)
\end{bmatrix} \nonumber\\
& \phantom{=} -\begin{bmatrix}
-\eta c\,H-\eta c\rho\,R^\top\Gamma_2 R+\eta^2\,R^\top\Gamma_1 R
& \hspace{3mm} &
-\eta c\,R^\top\Gamma_1+\dfrac{\eta^2}{\rho}\,R^\top(\Gamma_1-\mathbf I)
\\[0mm]
\hspace{0mm} & & \hspace{0mm}
\\[0mm]
-\eta RH-\eta\rho\,R R^\top\Gamma_2 R+c\eta\,\Gamma_1 R
& \hspace{3mm} &
-\eta\,R R^\top\Gamma_1+\dfrac{c\eta}{\rho}(\Gamma_1-\mathbf I) \end{bmatrix} \nonumber\\
& \phantom{=} -\underbrace{\frac{\eta\kappa_1}{2c}}_{\tau_2} \begin{bmatrix}
\eta c\, \mathbf{I} & \hspace{3mm} & \eta R^{\top} \\[0mm]
\hspace{0mm}        &              & \hspace{0mm}   \\[0mm]
\eta R              & \hspace{3mm} & c\, \mathbf{I}
\end{bmatrix} \nonumber\\
& = \begin{bmatrix}
       2\eta c\,H
+ 2\eta c\rho R^\top \Gamma_2 R
- 2\eta^2 R^\top \Gamma_1 R
- \frac{\eta^2\kappa_1}{2}\,\mathbf{I} & \eta\Bigl(
{
H +
{\rho R^\top \Gamma_2 R - \frac{\eta\kappa_1}{2c}\mathbf{I}}
}
\Bigr) R^\top
+
\eta{\frac{\eta}{\rho}R^\top(\mathbf{I}-\Gamma_1)}\\
*^\top & \eta(\Gamma_1 R R^\top + R R^\top\Gamma_1)
+ \frac{2\eta c}{\rho}(\mathbf{I}-\Gamma_1)
- \frac{\eta\kappa_1}{2}\,\mathbf{I}
  \end{bmatrix},
\end{align}
as stated.
\end{proof}

The following proposition is key to establish that $Q$ in \eqref{eq:explicit_form_Q} is positive semi-definite.
\begin{proposition}\label{prop:technical_lemma_na_li}{(\cite[Lemma~6]{qu_li_augPD})}
Let $c$ be chosen as in~\eqref{eq:ChoiceOfC_all}. Recall the diagonal function matrix $\Gamma_1$ from~\eqref{eq:Gamma_1and_Gamma_2} (see also~\eqref{eq:properties_gamma_i}).
Then, for $\eta>0$ and $\rho>0$ as introduced in~\eqref{eq:primal_dual_NON_EXPLICIT} and~\eqref{eq:Hrho_FTC}, respectively, the following inequality holds:
\begin{equation}\label{eq:lower_bound_Q2}
\eta(\Gamma_1 R R^\top + R R^\top\Gamma_1)
+ \frac{2\eta c}{\rho}(\mathbf{I}-\Gamma_1)
\succeq 
\frac{3}{2}\eta R R^\top,
\end{equation}
where $R$ is defined in~\eqref{eq:R_stacked}.
\hfill $\square$
\end{proposition}

\begin{proof}
Let us denote by $d_{1,1},d_{1,2},...,d_{1,n_\mathrm{c}}$ the main diagonal elements of $\Gamma_1=\operatorname{blkdiag}(\mathbf{I},\Gamma(\phi))$, with $\Gamma$ in \eqref{eq:diag_mat_GAMMA} and $\phi$ defined by components in \eqref{eq:phi_i}, and let us refer by $M(d_{1,1},...,d_{1,n_c})$ to  the  matrix in the left-hand side of \eqref{eq:lower_bound_Q2}:
\begin{equation}\label{eq:matrix_M_gamma}
M(d_{1,1},...,d_{1,n_c})
:= \eta\big(\Gamma_1 R R^\top + R R^\top \Gamma_1\big)
   + \frac{2\eta c}{\rho}\big(\mathbf{I} - \Gamma_1\big),
\end{equation}
where $\eta>0$, $\rho>0$ and $c > 0$ are fixed constants. Now, since  $\Gamma_1$ is diagonal, it holds that each entry of
$\Gamma_1 R R^\top$, $R R^\top \Gamma_1$, and $\mathbf{I} - \Gamma_1$
depends linearly on each $d_{1,i}$.
Then, for fixed $R, \eta, c$ and $\rho$, the mapping $
(d_{1,1},\dots,d_{1,n_c}) \mapsto M_\gamma(d_{1,1},\dots,d_{1,n_c})
$
is affine in each variable $d_{1,i}$.   Moreover, note that $d_{1,i}$ is either 1, or equal to some function $\gamma_j(\phi_j)$ introduced in \eqref{eq:gamma_i}. Then,   each   $d_{1,i}$ satisfies 
\begin{equation}\label{eq:constraints_d_js}
    d_{1,i}\in[0,1].
\end{equation}
These facts will be exploited to establish the following inequality:
\begin{equation}\label{eq:aux_low_bound_M_ONE}
M(d_{1,1},\dots,d_{1,n_c}) \succeq \frac{3}{2}\,\eta R R^\top,
\quad \forall d_{1,i} \in [0,1],
\end{equation}
from which \eqref{eq:lower_bound_Q2} would directly follow. For that purpose, note that since $M$ is affine in $d_{1,1},...,d_{1,n_c}$, and the unit cube $[0,1]^{n_c}$ is a convex polytope, it holds that $M$ can be written as a convex combination of $2^{n_c}$ constant matrices  obtained by evaluating $M$ at each  vertex $b = (b_1,\dots,b_{n_c}) \in \{0,1\}^{n_c}$ of $[0,1]^{n_c}$. That is, 
\begin{subequations}
\begin{equation}
M(d_{1,1},\dots,b_{1,n_c})
= \sum_{b \in \{0,1\}^{n_c}}
   \alpha_b(d_{1,1},...,d_{1,n_c})\, M(b_1,\dots,b_{n_c}),
\end{equation}
with coefficients $\alpha_b$ given by
\begin{equation}
\alpha_b(d_{1,1},...,d_{1,n_c})
:= \prod_{i : b_i = 1} d_{1,i} \;\prod_{i : b_i = 0} (1 - d_{1,i})
\end{equation}
and satisfying
\begin{equation}
\alpha_b(d_{1,1},...,d_{1,n_c}) \ge 0,\quad \sum_{b \in \{0,1\}^{n_c}} \alpha_b(d_{1,1},...,d_{1,n_c}) = 1.
\end{equation}
\end{subequations}
Then, to establish \eqref{eq:aux_low_bound_M_ONE}, it suffices to show that
\begin{equation}\label{eq:simple_suff_cond_Mb}
M(b_1,\dots,b_{n_c}) \succeq \frac{3}{2}\,\eta R R^\top
\quad \forall b \in \{0,1\}^{n_c}.
\end{equation}
To that end,  let $b\in \{0,1\}^{n_c}$ be arbitrary, and  let us assume that each of the first $k$ main diagonal entries of $\Gamma_1$ are equal to one, and that the remaining ones are zero\footnote{This assumption is without loss of generality. Indeed, let $b\in \{0,1\}^{n_c}$ be arbitrary, and define $k=\sum_i b_i$. Now, choose a permutation matrix \cite{horn} $P$ such that $P \Gamma_{1} P^\top
= \Gamma_{1,k} := \operatorname{diag}({1,\dots,1_k},0,\dots,0).
$
If we define  
$
\widetilde{R R^\top} := P R R^\top P^\top
$
and use the identity $P^\top P = \mathbf{I}$, then we obtain that
$
P(\Gamma_{1} R R^\top)P^\top
= (P\Gamma_{1} P^\top)(P R R^\top P^\top)
= \Gamma_{1,k} \widetilde{R R^\top}
$, $
P(R R^\top\Gamma_{1})P^\top
= \widetilde{R R^\top}\Gamma_{1,k}$, and
$
P(\mathbf{I}-\Gamma_{1})P^\top = \mathbf{I} - \Gamma_{1,k} 
$. 
Then, the permuted matrix
$
\widetilde M(b) := P M(b) P^\top
$
satisfies $
\widetilde M(b)
= \eta(\Gamma_{1,k} \widetilde{R R^\top}
      + \widetilde{R R^\top}\Gamma_{1,k})
  + \frac{2\eta c}{\rho}(\mathbf{I} - \Gamma_{1,k}).$
Since orthogonal congruence preserves the semi-definiteness order \cite{horn}, we have that $
M_\gamma(b)\succeq \tfrac{3}{2}\eta R R^\top$  if and only if $
\widetilde M \succeq \tfrac{3}{2}\eta \widetilde{R R^\top}.$
Hence we may assume without loss of generality that the ones in $\Gamma_1$
occupy the first $k$ diagonal entries.}. Then, we can write the product $R R^\top $ as follows:
\begin{equation}\label{eq:RRtop_terms_SIGMA}
R R^\top =
\begin{bmatrix}
\Sigma_1 & \Sigma_3 \\
\Sigma_3^\top & \Sigma_2
\end{bmatrix},
\end{equation}
where $\Sigma_1\in\mathbb{R}^{k\times k}$,
$\Sigma_2\in\mathbb{R}^{({n_c}-k)\times({n_c}-k)}$,
$\Sigma_3\in\mathbb{R}^{k\times({n_c}-k)}$.
It follows that
\[
\Gamma_{1} R R^\top
= \begin{bmatrix}\Sigma_1 & \Sigma_3 \\ 0 & 0\end{bmatrix},\qquad
R R^\top \Gamma_{1}
= \begin{bmatrix}\Sigma_1 & 0 \\ \Sigma_3^\top & 0\end{bmatrix},\quad 
\mathbf{I}-\Gamma_{1} = \begin{bmatrix}0&0\\0&\mathbf{I}_{{n_c}-k}\end{bmatrix}.
\]
Then,
\begin{align}
M(b_1,...,b_{n_c}) & = M({1,\dots,1_k},0,\dots,0) \nonumber\\
& = \eta\left(\begin{bmatrix}
    \mathbf{I}_k & 0\\
    0 & \mathbf{0}_{n_c-k}
\end{bmatrix}\begin{bmatrix}
\Sigma_1 & \Sigma_3 \\
\Sigma_3^\top & \Sigma_2
\end{bmatrix}+\begin{bmatrix}
\Sigma_1 & \Sigma_3 \\
\Sigma_3^\top & \Sigma_2
\end{bmatrix}\begin{bmatrix}
    \mathbf{I}_k & 0\\
    0 & \mathbf{0}_{n_c-k}
\end{bmatrix} \right)+\frac{2\eta c}{\rho}\left(\begin{bmatrix}
    \mathbf{I}_k & 0\\
    0 & \mathbf{I}_{n_c-k}
\end{bmatrix}-\begin{bmatrix}
    \mathbf{I}_k & 0\\
    0 & \mathbf{0}_{n_c-k}
\end{bmatrix}  \right) \nonumber\\
& = \eta \left(\begin{bmatrix}
    \Sigma_1 & \Sigma_3\\
    0 & 0 
\end{bmatrix}+\begin{bmatrix}
    \Sigma_1 & 0\\
    \Sigma_3^\top & 0
\end{bmatrix} \right)+\frac{2\eta c}{\rho}\begin{bmatrix}
    \mathbf{0}_k & 0\\
    0 & \mathbf{I}_{n_c-k}
\end{bmatrix} \nonumber\\
& =
\begin{bmatrix}
2\eta\Sigma_1 & \eta\Sigma_3 \\
\eta\Sigma_3^\top &
\frac{2\eta c}{\rho} \mathbf{I}_{{n_c}-k}
\end{bmatrix}.
\end{align}
Consequently, \eqref{eq:simple_suff_cond_Mb} is equivalent to the following expression:
\begin{equation}\label{eq:ineq_simple_TWO}
\frac{\eta}{2} \begin{bmatrix}
    \Sigma_1 & -\Sigma_3\\
    -\Sigma_3^\top & \frac{4c}{\rho}\mathbf{I}_{n_c-k}-3\Sigma_2
\end{bmatrix}\succeq \textbf{0} .
\end{equation}
From Assumption~\ref{assump:BoundOnR}, $R R^\top \preceq \kappa_2 \mathbf{I}$, where $\kappa_2$ is in \eqref{eq:BoundOnR}. Considering \eqref{eq:RRtop_terms_SIGMA}, it holds that $\Sigma_2\preceq \kappa_2 \mathbf{I}_{n_c-k}$. Moreover, since  $c\geq \kappa_2 \rho $ (see \eqref{eq:ChoiceOfC2}), it follows that the sub-block $ \frac{4c}{\rho}\mathbf{I}_{n_c-k}-3\Sigma_2$ satisfies the following 
\begin{equation}
    \frac{4c}{\rho}\mathbf{I}_{n_c-k}-3\Sigma_2 \succeq \Sigma_2.
\end{equation}
Then, the matrix in the left-hand side of \eqref{eq:ineq_simple_TWO} can be lower bounded as follows:
\begin{equation}\label{eq:final_comparison_SIGMAS}
    \frac{\eta}{2} \begin{bmatrix}
    \Sigma_1 & -\Sigma_3\\
    -\Sigma_3^\top & \frac{4c}{\rho}\mathbf{I}_{n_c-k}-3\Sigma_2
\end{bmatrix}\succeq \frac{\eta}{2} \begin{bmatrix}
    \Sigma_1 & -\Sigma_3\\
    -\Sigma_3^\top & \Sigma_2
\end{bmatrix}.
\end{equation}
The matrix in the right-hand side of \eqref{eq:final_comparison_SIGMAS} is congruent to $RR^\top$ with orthogonal matrix $\begin{bmatrix}\mathbf{I}_k & 0 \\ 0 & -\mathbf{I}_{{n_c}-k}\end{bmatrix}$. From Assumption~\ref{assump:BoundOnR}, $RR^\top$ is positive definite.  Therefore, \eqref{eq:ineq_simple_TWO} and consequently  \eqref{eq:simple_suff_cond_Mb}, which in chain implies \eqref{eq:aux_low_bound_M_ONE} and \eqref{eq:lower_bound_Q2},  hold true. This concludes the proof.

\end{proof}

\noindent \emph{Lemma~2:} Fix $\eta>0$ and $\epsilon\in(0,1)$,  let Assumptions~\ref{assump:uConvex} and \ref{assump:BoundOnR} hold, and  fix $c>0$ according to~\eqref{eq:ChoiceOfC_all}. Then,  $P_2$ in \eqref{eq:matrix_P_2} satisfies
\begin{equation}\label{eq:lower_bound_P_2}
P_2 \;\succeq\; \epsilon\,c\,\min(\eta,1)\,\mathbf I.
\end{equation}
Moreover, the \gls{a-pdgd}~\eqref{eq:PdControllerWithInput} is strictly shifted-passive   with quadratic storage function \(S_{\mathrm{c}}(\tilde{\theta})=\tilde{\theta}^{\top}P_2\tilde{\theta}\) and with respect to the input--output pair
\((\tilde{v}_{\mathrm{pd}},\, 2\,B_{\mathrm{pd}}^{\top} P_2 \tilde{\theta})\).
In particular,  it holds that
\[
\dot{S}_{\mathrm{c}}
\le -\tau_2\, \tilde{\theta}^{\top} P_2\, \tilde{\theta}
+ 2\,\tilde{\theta}^{\top} P_2 B_{\mathrm{pd}} \tilde{v}_{\mathrm{pd}}
\]
along any system trajectory, with 
\begin{align*}
\tau_2 &\coloneqq \frac{\eta \kappa_1}{2c}.
%\label{eq:tau2_def}    
\end{align*}
\hfill $\square$
\begin{proof}
Recall from Proposition~\ref{prop:lower_bound_P2} that $P_2 \succ \mathbf{0}$. Moreover, by Proposition~\ref{prop:shifted_primal_dual}, the Aug-PDGD \eqref{eq:PdControllerWithInput} is equivalent to its shifted dynamics \eqref{eq:shiftedPd-short}. Then, along any of its solutions the time derivative of $S_{\mathrm{c}}$ satisfies the following:
\begin{align}\label{eq:Lie_S-short}
\dot S_{\mathrm c}
&= \dot{\tilde{\theta}}^\top P_2 \tilde{\theta}
+ \tilde{\theta}^\top P_2 \dot{\tilde{\theta}} \nonumber\\
 & = \tilde{\theta}^\top (F^\top P_2 + P_2 F)\tilde{\theta}
+ 2\,\tilde{\theta}^\top P_2 B_{\mathrm{pd}}\tilde v_{\mathrm{pd}},
\end{align}
where for simplicity we have removed the dependency of $F$ with respect to $\theta$ (see \eqref{eq:F_theta}). To show shifted-passivity, it suffices to verify that
\[
F^\top P_2 + P_2 F \preceq -\tau_2 P_2,
\]
or, equivalently, that
\[
Q = -F^\top P_2 - P_2 F - \tau_2 P_2 \succeq \mathbf{0}.
\]
We recall from Proposition~\ref{prop:explicity_form_Q} that
\begin{align*}
Q & = \begin{bmatrix}
       2\eta c\,H
+ 2\eta c\rho R^\top \Gamma_2 R
- 2\eta^2 R^\top \Gamma_1 R
- \frac{\eta^2\kappa_1}{2}\,\mathbf{I} & \eta\Bigl(
{
H +
{\rho R^\top \Gamma_2 R - \frac{\eta\kappa_1}{2c}\mathbf{I}}
}
\Bigr) R^\top
+
\eta{\frac{\eta}{\rho}R^\top(\mathbf{I}-\Gamma_1)}\\
*^\top & \eta(\Gamma_1 R R^\top + R R^\top\Gamma_1)
+ \frac{2\eta c}{\rho}(\mathbf{I}-\Gamma_1)
- \frac{\eta\kappa_1}{2}\,\mathbf{I}
  \end{bmatrix}\\
  & =: \begin{bmatrix}
      Q_1 & Q_3\\
      Q_3^\top & Q_2
  \end{bmatrix}.
\end{align*}
By the Schur complement, to show that  \(Q\succeq \mathbf{0}\) it suffices to verify that
\[
Q_2\succ \mathbf{0},
\qquad 
Q_1 - Q_3 Q_2^{-1} Q_3^\top \succeq \mathbf{0}.
\]
From Proposition~\ref{prop:technical_lemma_na_li}, the condition \(c\ge \kappa_2 \rho\) is sufficient for  ensuring that
\begin{equation}\label{eq:aux_low_bound_Q2}
Q_2 
\succeq 
\frac{3}{2}\eta R R^\top - \frac{\eta\kappa_1}{2}\mathbf{I}.
\end{equation}
Since $ R R^\top\succeq \kappa_1 \mathbf{I}$ due to Assumption~\ref{assump:BoundOnR}, it holds that $Q_2\succ \textbf{0}$. 

We proceed to show that $Q_1 - Q_3 Q_2^{-1} Q_3^\top \succeq \mathbf{0}$.  To establish this, we will  first  show that the following holds:
\begin{subequations}
    \begin{align}
Q_1
& \succeq 
2\eta c H
-
h_1(c)
\mathbf{I} \label{eq:to_establish_low_bound_Q1}\\
Q_3 Q_2^{-1}Q_3^\top  & \preceq \eta \ell H+ (h_2(c)+h_3(c)+h_4(c))\mathbf{I}, \label{eq:to_establish_q3q2q3}
    \end{align}
    where 
    \begin{align}
h_1 & = 2\eta^2\kappa_2 + \tfrac{\eta^2\kappa_1}{2},\\
h_2(c) & = 2\eta \ell (\rho\kappa_2 + \frac{\eta\kappa_1}{2c})
+
\eta(\rho\kappa_2 + \frac{\eta\kappa_1}{2c})^2,\\
h_3(c) & =  \frac{2\eta^2}{\rho}
\Bigl(
\ell + \rho\kappa_2 + \frac{\eta\kappa_1}{2c}
\Bigr)
\frac{\kappa_2}{\kappa_1},\\
h_4 & =  \eta^3 \rho^{-2}
\frac{\kappa_2}{\kappa_1}.
    \end{align}
\end{subequations}

Before proceeding, we find it convenient to recall the following facts. The matrix $\Gamma_1$, defined in \eqref{eq:Gamma_1and_Gamma_2}, satisfies $\Gamma_1 \preceq \mathbf{I}$. The matrix $H$, defined in Proposition~\ref{prop:shifted_primal_dual}, is symmetric positive-definite and satisfies $\mu \mathbf{I}\preceq H \preceq \ell \mathbf{I}$. The matrix $R$, introduced in \eqref{eq:R_block_compact}, is such that $\kappa_1 \mathbf{I}\preceq RR^\top \preceq \kappa_2 \mathbf{I} $.

\textbf{Let us continue then by verifying \eqref{eq:to_establish_low_bound_Q1}}. Note directly that
\begin{equation}
    Q_1=  2\eta c\,H
+ 2\eta c\rho R^\top \Gamma_2 R
- 2\eta^2 R^\top \Gamma_1 R
- \frac{\eta^2\kappa_1}{2}\,\mathbf{I} \succeq 2\eta cH-2\eta^2R^\top \Gamma_1 R-\frac{\eta^2\kappa_1}{2}\mathbf{I}.
\end{equation}
Since $\Gamma_1 \preceq \mathbf{I}$ and $R^\top R \preceq \kappa_2 \mathbf{I}$, then
\eqref{eq:to_establish_low_bound_Q1} is verified.

\textbf{Next, we show that \eqref{eq:to_establish_q3q2q3} holds.} Let us write $Q_3$ as follows
\begin{equation}
    Q_3=\eta Q_{3,\mathrm{a}} R^\top + \eta Q_{3,\mathrm{b}},
\end{equation}    
where 
\begin{equation}\label{eq:Q3a_Q3b}
Q_{3,\mathrm{a}}  := 
H +
{\rho R^\top \Gamma_2 R - \frac{\eta\kappa_1}{2c}\mathbf{I}},\quad
Q_{3,\mathrm{b}} := {\frac{\eta}{\rho}R^\top(\mathbf{I}-\Gamma_1)}.
\end{equation}
Now, note from \eqref{eq:aux_low_bound_Q2} that  $Q_2\succeq \eta RR^\top$ implies that $Q_2^{-1}\preceq  \frac{1}{\eta}(RR^\top)^{-1}$ (recall that $RR^\top$ is positive-definite and hence invertible).  
Then the following holds:
\begin{align}\label{eq:bound_q3q2q3_ONE}
    Q_3 Q_2^{-1} Q_3^\top & \preceq Q_3 \left( \frac{1}{\eta}(RR^\top)^{-1} \right)Q_3^\top \nonumber\\
     & = (\eta Q_{3,\mathrm{a}} R^\top + \eta Q_{3,\mathrm{b}}) \left( \frac{1}{\eta}(RR^\top)^{-1} \right)(\eta R Q_{3,\mathrm{a}}^\top +\eta Q_{3,\mathrm{b}}^\top)\nonumber\\
     & = \eta Q_{3,\mathrm{a}}R^\top(RR^\top)^{-1}RQ_{3,\mathrm{a}}^\top+ 2\mathrm{symm}\{\eta Q_{3,\mathrm{a}}R^\top(RR^\top)^{-1}Q_{3,\mathrm{b}}^\top\} +\eta Q_{3,\mathrm{b}}(RR^\top)^{-1}Q_{3,\mathrm{b}}^\top.
\end{align}
Since for any symmetric matrix $M$, \( M\preceq \|M\|\mathbf{I}\), and \(R^\top (R R^\top)^{-1} R \preceq \mathbf{I}\), then the right-hand side of \eqref{eq:bound_q3q2q3_ONE} satisfies the following:
\begin{align}\label{eq:bound_q3q2q3_TWO}
&\eta Q_{3,\mathrm{a}}R^\top(RR^\top)^{-1}RQ_{3,\mathrm{a}}^\top+ 2\mathrm{symm}\{\eta Q_{3,\mathrm{a}}R^\top(RR^\top)^{-1}Q_{3,\mathrm{b}}^\top\} +\eta Q_{3,\mathrm{b}}(RR^\top)^{-1}Q_{3,\mathrm{b}}^\top  \nonumber\\
     & \preceq  
     \eta Q_{3,\mathrm{a}} Q_{3,\mathrm{a}}^\top 
+ 2\eta\bigl\| Q_{3,\mathrm{a}} R^\top (R R^\top)^{-1} Q_{3,\mathrm{b}}^\top \bigr\|\mathbf{I} + \eta\bigl\| Q_{3,\mathrm{b}} (R R^\top)^{-1} Q_{3,\mathrm{b}}^\top \bigr\|\mathbf{I}.
\end{align}
Let us develop upper bounds for each of the terms in the right-hand side of \eqref{eq:bound_q3q2q3_TWO}. We begin with $\eta Q_{3,\mathrm{a}} Q_{3,\mathrm{a}}^\top$:
\begin{align}\label{eq:bound_q3q2q3_THREE}
   \eta Q_{3,\mathrm{a}} Q_{3,\mathrm{a}}^\top & = \eta \left(H +
{\rho R^\top \Gamma_2 R - \frac{\eta\kappa_1}{2c}\mathbf{I}} \right)\left(H+\rho R^\top\Gamma_2R-\frac{\eta\kappa_1}{2c}\mathbf{I} \right) \nonumber\\
& = \eta \left(H^2+2\mathrm{symm\{}H(\rho R^\top \Gamma_2 R-\frac{\eta\kappa_1}{2c}\mathbf{I})\}+({\rho R^\top \Gamma_2 R - \frac{\eta\kappa_1}{2c}\mathbf{I}})^2 \right) \nonumber\\
&\preceq 
\eta H^2
+ 2\eta \|H\| \|\rho R^\top \Gamma_2 R - \frac{\eta\kappa_1}{2c}\mathbf{I}\| \mathbf{I}
+ \eta \|\rho R^\top \Gamma_2 R - \frac{\eta\kappa_1}{2c}\mathbf{I}\|^2 \mathbf{I} \nonumber
\\
&\preceq
\eta \ell H
+ \underbrace{\Bigl(
2\eta \ell (\rho\kappa_2 + \frac{\eta\kappa_1}{2c})
+
\eta(\rho\kappa_2 + \frac{\eta\kappa_1}{2c})^2
\Bigr)}_{\mathrm{h}_2(c)} \mathbf{I} \nonumber\\
& = \eta \ell H +h_2(c)\mathbf{I}.
\end{align}
We proceed analogously with the second term and third terms  in the right-hand side of \eqref{eq:bound_q3q2q3_TWO}. For the second term, using sub-multiplicativity, we have that

\begin{align}\label{eq:bound_q3q2q3_FOUR}
2\eta\bigl\| Q_{3,\mathrm{a}} R^\top (R R^\top)^{-1} Q_{3,\mathrm{b}}^\top \bigr\|\mathbf{I}
& \preceq 
2\eta\|Q_{3,\mathrm{a}}\|\,\|R^\top\|\|(R R^\top)^{-1}\|\|Q_{3,\mathrm{b}}^\top\|\mathbf{I} \notag \\
& = 
2\eta\Bigl\|H + \rho R^\top \Gamma_2 R - \frac{\eta\kappa_1}{2c}\mathbf{I}\Bigr\|\|R^\top\|\|(R R^\top)^{-1}\|
\Bigl\|\frac{\eta}{\rho}(\mathbf{I}-\Gamma_1)R\Bigr\|\mathbf{I} \notag \\
& \preceq 
2\eta\,\Bigl(\|H\| + \rho\,\|R^\top \Gamma_2 R\| + \frac{\eta\kappa_1}{2c}\|\mathbf{I}\|\Bigr)
\|R^\top\|\|(R R^\top)^{-1}\|
\Bigl\|\frac{\eta}{\rho}(\mathbf{I}-\Gamma_1)R\Bigr\|\,\mathbf{I} 
\end{align}
Since \eqref{eq:bounds_H_xi} implies $\|H\|\le \ell$, \eqref{eq:BoundOnR} implies $\|R^\top R\|\le \kappa_2$, $\|(R^\top R)^{-1}\|\le \kappa_1$, and $\|R\|=\|R^\top\|\le \sqrt{\kappa_2}$, and \eqref{eq:Gamma_1and_Gamma_2} implies $\|\Gamma_1\|\le 1$ and $\|\Gamma_2\|\le 1$, then, 
\begin{align}\label{eq:bound_q3q2q3_FOUR_CONTINUED}
2\eta\bigl\| Q_{3,\mathrm{a}} R^\top (R R^\top)^{-1} Q_{3,\mathrm{b}}^\top \bigr\|\mathbf{I} & \preceq  
2\eta\Bigl(\ell + \rho\kappa_2 + \frac{\eta\kappa_1}{2c}\Bigr)
\Bigl( \frac{\sqrt{\kappa_2}}{\kappa_1} \Bigr)
\Bigl(\frac{\eta \sqrt{\kappa_2}}{\rho} \Bigr)\mathbf{I}
\notag \\
& \preceq \left(\frac{2\eta^2}{\rho}
\Bigl(
\ell + \rho\kappa_2 + \frac{\eta\kappa_1}{2c}
\Bigr)
\frac{\kappa_2}{\kappa_1}\right)\mathbf{I} \notag\\
& =\mathrm{h}_3(c)\mathbf{I}.
\end{align}

We proceed analogously with the third term, to obtain:
\begin{align}\label{eq:bound_q3q2q3_FIVE}
\eta\bigl\| Q_{3,\mathrm{b}} (R R^\top)^{-1} Q_{3,\mathrm{b}}^\top \bigr\|\mathbf{I}
&\preceq \eta\,\|Q_{3,\mathrm b}\|\,\|(RR^\top)^{-1}\|\,\|Q^\top_{3,\mathrm b}\|\,\mathbf I \notag\\
&= \eta\,\|Q_{3,\mathrm b}\|^2\,\|(RR^\top)^{-1}\|\,\mathbf I \notag\\
&= \eta\,\Bigl\|\frac{\eta}{\rho}(\mathbf I-\Gamma_1)R\Bigr\|^2\,\|(RR^\top)^{-1}\|\,\mathbf I \notag \\
& \preceq \eta\,\Bigl(\frac{\eta}{\rho}\Bigr)^2\,\|\mathbf I-\Gamma_1\|^2\,\|R\|^2\,\|(RR^\top)^{-1}\|\,\mathbf I.
\end{align}
Since \eqref{eq:Gamma_1and_Gamma_2} implies $\| \mathbf{I}-\Gamma_1\| \leq 1$, \eqref{eq:BoundOnR} implies $\| R \| \leq \sqrt{\kappa_2}$ and $\| (R^\top R)^{-1}  \| \leq \kappa_1$, then, 
\begin{align}\label{eq:bound_q3q2q3_FIVE_CONTINUED}
\eta\bigl\| Q_{3,\mathrm{b}} (R R^\top)^{-1} Q_{3,\mathrm{b}}^\top \bigr\|\mathbf{I}&
\preceq \eta\,\Bigl(\frac{\eta}{\rho}\Bigr)^2\,(\kappa_2)\,\Bigl(\frac{1}{\kappa_1}\Bigr)\mathbf I \notag\\
&= \left(\eta^3\rho^{-2}\frac{\kappa_2}{\kappa_1}\right)\mathbf I \notag \\
& = \mathrm h_4(c)\,\mathbf I.
\end{align}

\textbf{By considering \eqref{eq:bound_q3q2q3_ONE}--\eqref{eq:bound_q3q2q3_FIVE_CONTINUED}, it follows that
\eqref{eq:to_establish_q3q2q3} holds.}  Then we can move on and combine 
\eqref{eq:to_establish_low_bound_Q1} and \eqref{eq:to_establish_q3q2q3} to obtain the following:
\begin{align}
Q_1 - Q_3 Q_2^{-1} Q_3^\top
&\succeq
(2\eta c - \eta\ell) H
-
\bigl(
\mathrm{h}_1
+ \mathrm{h}_2(c)
+ \mathrm{h}_3(c)
+ \mathrm{h}_4
\bigr)\mathbf{I}
\nonumber\\[1mm]
&\succeq
\Bigl[
2\eta\mu c
-
\bigl(
\eta\mu\ell
+
\mathrm{h}_1+\mathrm{h}_2(c)+\mathrm{h}_3(c)+\mathrm{h}_4
\bigr)
\Bigr]\mathbf{I},
\label{eq:schur_lower_bound}
\end{align}
where we again have used $\mu \preceq H$ from \eqref{eq:bounds_H_xi}.  Let
\begin{align}
\Delta(c)
\coloneqq\
2\eta\mu c
-
\bigl(
\eta\mu\ell
+
\mathrm{h}_1+\mathrm{h}_2(c)+\mathrm{h}_3(c)+\mathrm{h}_4
\bigr),
\label{eq:Delta_def}
\end{align}
then \eqref{eq:schur_lower_bound} is implied from
\(
Q_1 - Q_3 Q_2^{-1} Q_3^\top \succeq \Delta(c)\mathbf{I}.
\)
We now show that $\Delta(c_o)>0$ for the choice $c=c_o$, where $c_o$ denotes the smallest admissible value satisfying \eqref{eq:ChoiceOfC3}, i.e., $c_0 \ge
20\,\ell\,
\Big[\max\!\left(\tfrac{\rho \kappa_{2}}{\mu},\, \tfrac{\ell}{\mu}\right)\Big]^{2}
\Big[\max\!\left(\tfrac{\eta}{\ell\rho},\, \tfrac{\ell}{\mu}\right)\Big]^{2}
\tfrac{\kappa_{2}}{\kappa_{1}}$. To this end, define
\begin{subequations}\label{eq:q_defs}
\begin{align}
q_{\rho}
&\coloneqq
\max\left(\frac{\rho\kappa_{2}}{\mu},\frac{\ell}{\mu}\right),
\\
q_{\eta}
&\coloneqq
\max\left(\frac{\eta}{\ell\rho},\,\frac{\ell}{\mu}\right),
\end{align}
\end{subequations}
and set $c=c_o$, where we take equality in \eqref{eq:ChoiceOfC3}, i.e.,
\begin{align}
c_o
\coloneqq
20\ell q_{\rho}^{2} q_{\eta}^{2}\frac{\kappa_{2}}{\kappa_{1}}.
\label{eq:co_def}
\end{align}
By construction, the following inequalities hold:
\begin{align}
\rho\kappa_2 \le q_\rho \mu,
\qquad
\ell \le q_\rho \mu,
\qquad
\eta \le q_\eta \ell\rho,
\qquad
q_\rho \ge 1,
\qquad
q_\eta \ge 1:
\label{eq:q_properties}
\end{align}
We note in particular that $q_\rho\geq 1$ and $q_\eta\geq 1$ due to the fact that $\mu\leq \ell\Leftrightarrow 1\leq \ell/\mu$ according to Proposition~\ref{prop:shifted_primal_dual} and Assumption~\ref{assump:uConvex}.
Substituting \eqref{eq:co_def} into $2\eta\mu c$ yields
\begin{align}
2\eta\mu c_o
=
40\eta\mu\ell q_\rho^2 q_\eta^2
\frac{\kappa_2}{\kappa_1}.
\label{eq:2etamuco}
\end{align}
\textbf{Next, we will produce a lower bound for  $\Delta(c_o)$ by finding bounds for each of its composing terms.} Note that since $\kappa_2/\kappa_1\geq 1$, due to Assumption~\ref{assump:BoundOnR}, then
$q_\rho^2 q_\eta^2 \frac{\kappa_2}{\kappa_1} \ge 1$. As a consequence,  
\eqref{eq:2etamuco} implies that
\begin{align}
\eta\mu\ell
\le
\frac{1}{40}\,(2\eta\mu c_o).
\label{eq:bound_etamul}
\end{align}

Next, recall that
\begin{align}
\mathrm{h}_1
=
2\eta^2\kappa_2
+
\frac{\eta^2\kappa_1}{2}.
\label{eq:h1_def}
\end{align}
Since $\kappa_1 \le \kappa_2$ (by Assumption~\ref{assump:BoundOnR}), it follows that
\begin{align}
\mathrm{h}_1
\le
\frac{5}{2}\eta^2\kappa_2.
\label{eq:h1_step1}
\end{align}
Using $\eta \le q_\eta\ell\rho$ and $\rho\kappa_2 \le q_\rho\mu$ from \eqref{eq:q_properties}, we further obtain that
\begin{align}
\mathrm{h}_1
\le
\frac{5}{2}\eta(q_\eta\ell\rho)\kappa_2
\le
\frac{5}{2}(\eta \mu   \ell)  q_\eta q_\rho.
\label{eq:h1_step2}
\end{align}
From \eqref{eq:2etamuco}, $\eta \mu\ell=\tfrac{1}{20}\tfrac{\eta \mu c_0}{q_\rho^2 q_\eta ^2\tfrac{\kappa_2}{\kappa_1}}$. As noted earlier,  $q_\rho,q_\eta\ge 1$. Also, $\kappa_2/\kappa_1\geq 1$ due to Assumption~\ref{assump:BoundOnR}. Then, the following is implied from \eqref{eq:h1_step2}:
\begin{align}
\mathrm{h}_1
\le
\frac{1}{16}\,(2\eta\mu c_o).
\label{eq:bound_h1}
\end{align}

For $\mathrm{h}_2(c)$, recall that 
\begin{align}
\mathrm{h}_2(c_o) & =2\eta \ell (\rho\kappa_2 + \frac{\eta\kappa_1}{2c_0})
+
\eta(\rho\kappa_2 + \frac{\eta\kappa_1}{2c_0})^2.
%  & \le FROM Here \\
% &  \le
% (2\eta\mu c_o)\,
% \frac{\kappa_1}{\kappa_2}
% \left(
% \frac{41}{800\,q_\rho q_\eta}
% +
% \frac{1681\,\mu}{64000\,\ell}
% \right).
\label{eq:h2_intermediate}
\end{align}
Using $c_o=20\ell q_{\rho}^{2} q_{\eta}^{2}\frac{\kappa_{2}}{\kappa_{1}}$ from \eqref{eq:co_def} and $\rho \kappa_2 \le q_\rho \mu$ from \eqref{eq:q_properties}, we get from \eqref{eq:h2_intermediate} that
\begin{align*}
   \mathrm{h}_2(c_o) \le 2\eta \ell ( q_\rho \mu + \frac{\eta\kappa^2_1}{40 \ell q^2_\rho q^2_\eta \kappa_2})
+
\eta( q_\rho \mu + \frac{\eta\kappa^2_1}{40 \ell q^2_\rho q^2_\eta \kappa_2})^2,
\end{align*}
Since $q_\rho , q_\eta \ge 1$ (see \eqref{eq:q_properties}) and $\kappa_2 \ge \kappa_1$ (see Assumption \ref{assump:BoundOnR}), it follows that
\begin{align*}
    \mathrm{h}_2(c_o) \le 2\eta \ell ( q_\rho \mu + \frac{\eta\kappa_1}{40 \ell })
+
\eta( q_\rho \mu + \frac{\eta\kappa_1}{40 \ell })^2.
\end{align*}
Given that $\eta \le q_{\eta}\ell\rho$ and $q_\eta \ge 1$ (see \eqref{eq:q_properties}), the above expression implies that
\begin{align*}
   \mathrm{h}_2(c_o) \le 2\eta \ell ( q_\eta q_\rho \mu + \frac{q_{\eta}\,\ell\rho \kappa_1}{40 \ell })
+
\eta(q_\eta q_\rho \mu + \frac{q_{\eta}\,\ell\rho \kappa_1}{40 \ell })^2.
\end{align*}
Using again the facts that $\kappa_2 \ge \kappa_1$ and $\rho\kappa_{2} \le q_{\rho}\,\mu$, we further get that
\begin{align*}
    \mathrm{h}_2(c_o) \le 2\eta \ell (  \frac{41q_{\eta}\, q_{\rho} \mu }{40 })
+
\eta(\frac{41q_{\eta}\, q_{\rho} \mu }{40 })^2.
\end{align*}
From  \eqref{eq:2etamuco}, $q_\rho q_\eta\mu =\frac{2\eta \mu c_0}{40\eta \mu \ell q_\rho q_\eta \tfrac{\kappa_1}{\kappa_2}}$, then
\begin{align*}
    \mathrm{h}_2(c_o) & \le 2\eta \ell (  \frac{41q_{\eta}\, q_{\rho} \mu }{40 }) (\frac{2 \eta \mu c_o}{40 \eta \mu \ell  q^2_{\rho} q^2_{\eta} \frac{\kappa_2}{\kappa_1}})
+
\eta(\frac{41q_{\eta}\, q_{\rho} \mu }{40 })^2 (\frac{2 \eta \mu c_o}{40 \eta \mu \ell  q^2_{\rho} q^2_{\eta} \frac{\kappa_2}{\kappa_1}})\\
 & = (2\eta\mu c_o)\,\frac{\kappa_1}{\kappa_2}
\left(
\frac{41}{800\,q_\rho q_\eta}
+
\frac{1681\,\mu}{64000\,\ell}
\right).    
\end{align*}
Considering that $q_\rho,q_\eta\ge 1$, due to \eqref{eq:q_properties}, $\kappa_1/\kappa_2\le 1$, due to Assumption~\ref{assump:BoundOnR}, and $\mu/\ell\le 1$, due to Proposition~\ref{prop:shifted_primal_dual} and Assumption~\ref{assump:uConvex},
we  obtain the following inequality for $h_2(c)$:
\begin{align}
\mathrm{h}_2(c_o)
\le
\frac{4961}{64000}\,(2\eta\mu c_o).
\label{eq:bound_h2}
\end{align}

Now for $\mathrm{h}_3(c)$,  recall that
\begin{align}
\mathrm{h}_3(c)
=
\frac{2\eta^2}{\rho}
\Bigl(
\ell + \rho\kappa_2 + \frac{\eta\kappa_1}{2c}
\Bigr)
\frac{\kappa_2}{\kappa_1}.
\label{eq:h3_def}
\end{align}
Using $\ell \leq q_\rho \mu$ and $\rho \kappa_2\leq q_\rho \mu$ from \eqref{eq:q_properties} and  $c=c_o=20\ell q_{\rho}^{2} q_{\eta}^{2}\frac{\kappa_{2}}{\kappa_{1}}$ from \eqref{eq:co_def}, we obtain the following:
\begin{align*}
    \mathrm{h}_3(c_o) \le  \frac{2\eta^2}{\rho}
\Bigl(
2q_\rho \mu + \frac{\eta\kappa^2_1}{40 \ell q^2_\rho q^2_\eta \kappa_2}
\Bigr)
\frac{\kappa_2}{\kappa_1}.
\end{align*}
Since $q_\rho,q_\eta \ge 1$ (see \eqref{eq:q_properties}), $\kappa_2 \ge \kappa_1$ (see Assumption \ref{assump:BoundOnR}), it further follows that
\begin{align*}
   \mathrm{h}_3(c_o) \le  \frac{2\eta^2}{\rho}
\Bigl(
2q_\rho \mu +\frac{\eta\kappa_1}{40 \ell }
\Bigr)
\frac{\kappa_2}{\kappa_1}.
\end{align*}
From \eqref{eq:q_properties},  $\eta \le q_{\eta}\,\ell\rho$ and $q_\eta \ge 1$. Then,
\begin{align*}
  \mathrm{h}_3(c_o) \le  \frac{2\eta^2}{\rho}
\Bigl(
2q_\eta q_\rho \mu + \frac{q_{\eta}\,\ell\rho \kappa_1}{40 \ell }
\Bigr)
\frac{\kappa_2}{\kappa_1}.
\end{align*}
Using again the inequalities $\kappa_2 \ge \kappa_1$ (see Assumption \ref{assump:BoundOnR}) and $\rho \kappa_2 \le q_\rho \mu$ (see \eqref{eq:q_properties}), leads to
\begin{align*}
    \mathrm{h}_3(c_o) \le  \frac{2\eta^2}{\rho}
\Bigl(
 \frac{81 q_{\eta} q_\rho \mu}{40 }
\Bigr)
\frac{\kappa_2}{\kappa_1}.
\end{align*}
Given that $\eta \le q_{\eta}\,\ell\rho$ (see \eqref{eq:q_properties}), and using $q_\rho q_\eta\mu =\frac{2\eta \mu c_0}{40\eta \mu \ell q_\rho q_\eta \tfrac{\kappa_1}{\kappa_2}}$ from \eqref{eq:2etamuco},  then the following expression can be obtained:
\begin{align*}
    \mathrm{h}_3(c_o) \le  \frac{2\eta}{\rho} (q_\eta \ell \rho)
\Bigl(
 \frac{81 q_{\eta} q_\rho \mu}{40 }
\Bigr)
\frac{\kappa_2}{\kappa_1}  \frac{2 \eta \mu c_o}{40 \eta \mu \ell  q^2_{\rho} q^2_{\eta} \frac{\kappa_2}{\kappa_1}}= \frac{81}{800 q_\rho} (2 \eta \mu c_o).
\end{align*}
Since $q_\rho \ge 1$ (see \eqref{eq:q_properties}), we can  obtain a suitable upper bound for $h_3(c)$ as follows:
\begin{align}
    \mathrm{h}_3(c_o) \le \frac{81}{800} (2 \eta \mu c_o).
\label{eq:bound_h3}
\end{align}

Finally, recall that
\begin{align}
\mathrm{h}_4
=
\eta^3\rho^{-2}\frac{\kappa_2}{\kappa_1}.
\label{eq:h4_def}
\end{align}
From \eqref{eq:q_properties}, it holds that  $\frac{\eta}{\rho}  \le q_{\eta}\,\ell$. Then,
\begin{align*}
\mathrm{h}_4 & \le  \eta q^2_{\eta}\,\ell^2
\frac{\kappa_2}{\kappa_1}.
\end{align*}
Through \eqref{eq:2etamuco}, $\frac{2 \eta \mu c_o}{40 \eta \mu \ell  q^2_{\rho} q^2_{\eta} \frac{\kappa_2}{\kappa_1}}=1$, then
\begin{align*}
\mathrm{h}_4 & \le  \eta q^2_{\eta}\,\ell^2
\frac{\kappa_2}{\kappa_1} ( \frac{2 \eta \mu c_o}{40 \eta \mu \ell  q^2_{\rho} q^2_{\eta} \frac{\kappa_2}{\kappa_1}}),
\end{align*}
which after simplification leads to
\begin{align*}
\mathrm{h}_4 & \le \frac{\ell}{40 \mu q^2_\rho} (2 \eta \mu c_o).
\end{align*}
Since from  \eqref{eq:q_properties} it holds that $\frac{\ell}{\mu} \le q_\rho$ and $q_\rho \ge 1$, then
\begin{align}
\mathrm{h}_4 & \le \frac{1}{40} (2 \eta \mu c_o).
\label{eq:bound_h4}
\end{align}

Collecting the bounds \eqref{eq:bound_etamul},
\eqref{eq:bound_h1},
\eqref{eq:bound_h2},
\eqref{eq:bound_h3},
and \eqref{eq:bound_h4}, we obtain that
\begin{align}
\eta\mu\ell
+
\mathrm{h}_1
+
\mathrm{h}_2(c_o)
+
\mathrm{h}_3(c_o)
+
\mathrm{h}_4
& \le
\Bigl(
\frac{1}{40}
+
\frac{1}{16}
+
\frac{4961}{64000}
+
\frac{81}{800}
+
\frac{1}{40}
\Bigr)
(2\eta\mu c_o) \nonumber\\
& = 0.291 \left(2\eta \mu c_0 \right).
\label{eq:sum_bounds}
\end{align}
\textbf{With \eqref{eq:sum_bounds} we can produce a lower bound for $\Delta(c_0)$ as follows:}
\begin{align}
\Delta(c_o)
\ge
\bigl(1-0.291\bigr)(2\eta\mu c_o)
>
0.
\label{eq:Delta_positive}
\end{align}
Consequently, from \eqref{eq:schur_lower_bound},
\begin{align}
Q_1 - Q_3 Q_2^{-1} Q_3^\top \succ \mathbf{0}.
\label{eq:schur_pd}
\end{align}
Thus, the shifted \gls{a-pdgd} is strictly passive  with quadratic storage function \(S_{\mathrm{c}}(\tilde{\theta})=\tilde{\theta}^{\top}P_2\tilde{\theta}\) and with respect to the input--output pair
\((\tilde{v}_{\mathrm{pd}},\, 2\,B_{\mathrm{pd}}^{\top} P_2 \tilde{\theta})\).
\end{proof}

\end{document}

%% file: figs/fig_buses_tikz.tex
% fig_buses_tikz.tex

\begin{circuitikz}[american currents, scale=0.8, transform shape]

% -------- Voltage-following bus j ∈ N --------
    \draw (4,0) node[circ,label=above:$v_{\mathrm{f},j}$] (busF) {};

    \draw (2.5,-2) node[ground] (gF1) {}
          to[isource,l=$i_{\mathrm{f},j}$] (2.5,0) coordinate (iF_top);
    \draw (2.5,0) to[short] (busF);

    \draw (4,0) to[C,l=$c_{\mathrm{f},j}$] (4,-2) node[ground] (gF2) {};
    \draw (4,0) to[short] (busF);

    \draw (5.5,0) to[generic,l=$\psi_{\mathrm{load},j}$] (5.5,-2) node[ground] (gF3) {};
    \draw (5.5,0) to[short] (busF);

    \draw (7.75,0) to[isource,l_=$\delta_{j}$] (7.75,-2) node[ground] (gF4) {};
    \draw (7.75,0) to[short] (busF);

    \node[
        draw, rounded corners,
        inner xsep=22pt, inner ysep=23pt,
        fit=(iF_top)(gF1)(gF2)(gF3)(gF4)(busF),
        label={[anchor=north,yshift=-2pt]north:$\boldsymbol{j\in\mathcal{N}_{\mathrm{f},i}}$}
    ] {};

% -------- Voltage-setting bus k ∈ S_i --------
    \begin{scope}[xshift=6.1cm]
        \draw (4,0) node[circ,label=above:$v_{\mathrm{s},j}$] (busS) {};
        \draw (4,0) to[cvsource,l=$v_{\mathrm{s},j}$] (4,-2) node[ground] (gS) {};
        \path (4,0) coordinate (s_top);

        \node[
            draw, rounded corners,
            inner xsep=22pt, inner ysep=23pt,
            fit=(s_top)(busS)(gS),
            label={[anchor=north,yshift=-2pt]north:$\boldsymbol{j\in\mathcal{N}_{\mathrm{s},i}}$}
        ] {};
    \end{scope}

\end{circuitikz}

%% file: references.bib
@article{qu_li_augPD2,
  author       = {Qu, Guannan and Li, Na},
  title        = {On the Exponential Stability of Primal-Dual Gradient Dynamics},
  journal      = {arXiv preprint},
  year         = {2018},
  eprint       = {1803.01825},
  archivePrefix= {arXiv},
  primaryClass = {math.OC},
  url          = {https://arxiv.org/abs/1803.01825}
}

@article{ortega2004interconnection,
  title={Interconnection and damping assignment passivity-based control: A survey},
  author={Ortega, Romeo and Garcia-Canseco, Eloisa},
  journal={European Journal of control},
  volume={10},
  number={5},
  pages={432--450},
  year={2004},
  publisher={Elsevier}
}

@book{van2000l2,
  title={L2-gain and passivity techniques in nonlinear control},
  author={Van der Schaft, Arjan},
  year={2000},
  publisher={Springer}
}

@book{horn,
  title={Matrix analysis},
  author={Horn, Roger A and Johnson, Charles R},
  year={2012},
  publisher={Cambridge university press}
}

@book{liberzon2003,
  title={Switching in systems and control},
  author={Liberzon, Daniel},
  volume={190},
  year={2003},
  publisher={Springer}
}

@article{Cucuzzella,
	author = {{Cucuzzella, Michele} and {Scherpen, Jacquelien M. A.} and {Machado, Juan E.}},
	title = {Microgrids control: AC or DC, that is not the question},
	DOI= "10.1051/epjconf/202431000015",
	journal = {EPJ Web Conf.},
	year = 2024,
	volume = 310,
	pages = "00015",
}

@ARTICLE{Dall’Anese,
  author={Colombino, Marcello and Dall’Anese, Emiliano and Bernstein, Andrey},
  journal={IEEE Transactions on Control of Network Systems}, 
  title={Online optimization as a feedback controller: Stability and tracking}, 
  year={2020},
  volume={7},
  number={1},
  pages={422-432},
  doi={10.1109/TCNS.2019.2906916}}

@ARTICLE{qu_li_augPD,
  author  = {Qu, Guannan and Li, Na},
  journal = {IEEE Control Systems Letters},
  title   = {On the exponential stability of primal-dual gradient dynamics},
  year    = {2019},
  volume  = {3},
  number  = {1},
  pages   = {43--48},
  doi     = {10.1109/LCSYS.2018.2851375}
}

@ARTICLE{Huang,
  author  = {Huang, Mingyu and Ding, Li and Li, Wenqu and Chen, Chao-Yang and Liu, Zhiwei},
  journal = {IEEE Transactions on Circuits and Systems I: Regular Papers},
  title   = {Distributed observer-based $H_\infty$ fault-tolerant control for {DC} microgrids with sensor fault},
  year    = {2021},
  volume  = {68},
  number  = {4},
  pages   = {1659--1670},
  doi     = {10.1109/TCSI.2020.3048971}
}

@misc{zaidi2025dcmg,
  author       = {Syed, Wasif},
  title        = {DC Microgrid Model, Constraint and Control Parameters (v1.0.1)},
  year         = {2026},
  doi          = {10.5281/zenodo.18529516},
  url          = {https://doi.org/10.5281/zenodo.18529516},
  note         = {Dataset},
}

@book{khalil2002nonlinear,
  title={Nonlinear Systems},
  author={Khalil, H.K.},
  volume    = {3},
  year      = {2002},
  publisher={Prentice Hall}
}

@ARTICLE{opt-alg-robust-fb,
  title   = {Optimization algorithms as robust feedback controllers},
  journal = {Annual Reviews in Control},
  volume  = {57},
  pages   = {100941},
  year    = {2024},
  author  = {Hauswirth, Adrian and He, Zhiyu and Bolognani, Saverio and Hug, Gabriela and D{\"o}rfler, Florian},
  doi     = {10.1016/j.arcontrol.2024.100941}
}

@BOOK{boyd_book,
  title     = {Convex optimization},
  author    = {Boyd, Stephen and Vandenberghe, Lieven},
  year      = {2004},
  publisher = {Cambridge University Press}
}

@ARTICLE{ZHANG,
  author  = {Zhang, Youmin and Jiang, Jin},
  title   = {Bibliographical review on reconfigurable fault-tolerant control systems},
  journal = {Annual Reviews in Control},
  volume  = {32},
  number  = {2},
  pages   = {229--252},
  year    = {2008},
  doi     = {10.1016/j.arcontrol.2008.03.008}
}

@ARTICLE{ChunXie,
  author  = {Xie, Chun-Hua and Yang, Guang-Hong},
  title   = {Decentralized adaptive fault-tolerant control for large-scale systems with external disturbances and actuator faults},
  journal = {Automatica},
  volume  = {85},
  pages   = {83--90},
  year    = {2017},
  doi     = {10.1016/j.automatica.2017.07.037}
}

@INPROCEEDINGS{SchenkFrance,
  author    = {Schenk, Kai and Lunze, Jan},
  booktitle = {2020 28th Mediterranean Conference on Control and Automation (MED)},
  title     = {Fault-tolerant task allocation in networked control systems},
  year      = {2020},
  pages     = {313--318},
  doi       = {10.1109/MED48518.2020.9182977}
}

@INPROCEEDINGS{SchenkMorocco,
  author    = {Schenk, Kai and Lunze, Jan},
  booktitle = {2019 4th Conference on Control and Fault Tolerant Systems (SysTol)},
  title     = {Fault tolerance in networked systems through flexible task assignment},
  year      = {2019},
  pages     = {257--263},
  doi       = {10.1109/SYSTOL.2019.8864764}
}

@ARTICLE{SchenkCooperative,
  author  = {Schenk, Kai and G{\"u}lbitti, Baris and Lunze, Jan},
  title   = {Cooperative fault-tolerant control of networked control systems},
  journal = {IFAC-PapersOnLine},
  volume  = {51},
  number  = {24},
  pages   = {570--577},
  year    = {2018},
  note    = {10th IFAC Symposium on Fault Detection, Supervision and Safety for Technical Processes (SAFEPROCESS 2018)},
  doi     = {10.1016/j.ifacol.2018.09.633}
}

@ARTICLE{Francesca,
  author  = {Boem, Francesca and Gallo, Alexander J. and Raimondo, Davide M. and Parisini, Thomas},
  journal = {IEEE Transactions on Control of Network Systems},
  title   = {Distributed fault-tolerant control of large-scale systems: An active fault diagnosis approach},
  year    = {2020},
  volume  = {7},
  number  = {1},
  pages   = {288--301},
  doi     = {10.1109/TCNS.2019.2913557}
}

@BOOK{ferrari,
  title     = {Safety, security and privacy for cyber-physical systems},
  author    = {Ferrari, Riccardo M. G. and Teixeira, Andr{\'e} M. H.},
  year      = {2021},
  publisher = {Springer}
}

@ARTICLE{Xian-Ming,
  author  = {Zhang, Xian-Ming and Han, Qing-Long and Ge, Xiaohua and Ding, Derui and Ding, Lei and Yue, Dong and Peng, Chen},
  journal = {IEEE/CAA Journal of Automatica Sinica},
  title   = {Networked control systems: A survey of trends and techniques},
  year    = {2020},
  volume  = {7},
  number  = {1},
  pages   = {1--17},
  doi     = {10.1109/JAS.2019.1911651}
}

@ARTICLE{CBI_OrtegaArjan,
  author  = {Ortega, Romeo and van der Schaft, Arjan and Castanos, Fernando and Astolfi, Alessandro},
  journal = {IEEE Transactions on Automatic Control},
  title   = {Control by interconnection and standard passivity-based control of port-{Hamiltonian} systems},
  year    = {2008},
  volume  = {53},
  number  = {11},
  pages   = {2527--2542},
  doi     = {10.1109/TAC.2008.2006930}
}

@INPROCEEDINGS{michele,
  author    = {Kosaraju, Krishna Chaitanya and Cucuzzella, Michele and Scherpen, Jacquelien M. A.},
  booktitle = {2019 IEEE 58th Conference on Decision and Control (CDC)},
  title     = {Distributed control of {DC} microgrids using primal-dual dynamics},
  year      = {2019},
  pages     = {6215--6220},
  doi       = {10.1109/CDC40024.2019.9029175}
}

@ARTICLE{claudio_arjan,
  author  = {Stegink, Tjerk W. and De Persis, Claudio and van der Schaft, Arjan J.},
  title   = {Stabilization of structure-preserving power networks with market dynamics},
  journal = {IFAC-PapersOnLine},
  volume  = {50},
  number  = {1},
  pages   = {6737--6742},
  year    = {2017},
  note    = {20th IFAC World Congress},
  doi     = {10.1016/j.ifacol.2017.08.1172}
}
